\documentclass[11pt]{article}

\usepackage{hyperref}
\title{Prediction of functional ARMA processes with an application to traffic data}
\author{Johannes Klepsch\thanks{Center for Mathematical Sciences, 
Technical University of Munich,  85748 Garching,\newline
Boltzmannstrasse~3, Germany, 
e-mail: j.klepsch@tum.de\,,\,cklu@tum.de\,,\,taoran.wei@tum.de, \newline
http://www.statistics.tum.de}
\and 
Claudia Kl\"uppelberg\footnotemark[1]
\and
Taoran Wei\footnotemark[1]
}

\usepackage{amssymb}
\usepackage{amsfonts}
\usepackage{amsmath}
\usepackage{amsthm}
\usepackage{inputenc}
\usepackage{graphicx}
\usepackage{caption}
\usepackage{subcaption}
\usepackage[usenames, dvipsnames]{xcolor}
\usepackage{verbatim}
\usepackage{dsfont}
\usepackage{color}
\usepackage[numbers]{natbib}
\usepackage{relsize}
\usepackage{lmodern}
\usepackage{url}
\usepackage{siunitx}
\usepackage{MnSymbol}
\usepackage[english]{babel}
\usepackage{placeins}
\numberwithin{equation}{section}
\usepackage{caption} 
\newtheorem{theorem}{Theorem}[section]
\newtheorem{lemma}[theorem]{Lemma}
\newtheorem{remark}[theorem]{Remark}
\newtheorem{example}[theorem]{Example}
\newtheorem{proposition}[theorem]{Proposition}
\newtheorem{definition}[theorem]{Definition}
\newtheorem{assumption}[theorem]{Assumption}
\newtheorem{corollary}[theorem]{Corollary}
\newtheorem{fig}[theorem]{Figure}

\newcommand{\bthe}{\begin{theorem}}
\newcommand{\ethe}{\end{theorem}}
\newcommand{\ben}{\begin{enumerate}}
\newcommand{\een}{\end{enumerate}}
\newcommand{\bit}{\begin{itemize}}
\newcommand{\eit}{\end{itemize}}
\newcommand{\beq}{\begin{equation}}
\newcommand{\eeq}{\end{equation}}
\newcommand{\ble}{\begin{lemma}}
\newcommand{\ele}{\end{lemma}}
\newcommand{\bde}{\begin{definition}\rm}
\newcommand{\ede}{\halmos\end{definition}}
\newcommand{\bco}{\begin{corollary}}
\newcommand{\eco}{\end{corollary}}
\newcommand{\bpr}{\begin{proposition}}
\newcommand{\epr}{\end{proposition}}
\newcommand{\brem}{\begin{remark}\rm}
\newcommand{\erem}{\halmos\end{remark}}
\newcommand{\bproof}{\begin{proof}}
\newcommand{\eproof}{\end{proof}}
\newcommand{\bexam}{\begin{example}\rm}
\newcommand{\eexam}{\halmos\end{example}}
\newcommand{\bfi}{\begin{fig}}
\newcommand{\efi}{\end{fig}}
\newcommand{\btab}{\begin{tab}}
\newcommand{\etab}{\end{tab}}
\newcommand{\beao}{\begin{eqnarray*}}
\newcommand{\eeao}{\end{eqnarray*}\noindent}
\newcommand{\beam}{\begin{eqnarray}}
\newcommand{\eeam}{\end{eqnarray}\noindent}
\newcommand{\barr}{\begin{array}}
\newcommand{\earr}{\end{array}}
\newcommand{\bdis}{\begin{displaymath}}
\newcommand{\edis}{\end{displaymath}\noindent}

\def\N{{\mathbb N}}

\def\Z{{\mathbb Z}}

\def\R{{\mathbb R}}

\def\cals_+{{\cals_+}}

\def\call{{\mathcal{L}}}

\def\cals{{\mathcal{S}}}

\def\sp{{\overline{{{\rm sp}}}}}

\newcommand{\la}{{\lambda}}

\newcommand{\eps}{\varepsilon}

\newcommand{\ARMA}{{\mbox{\rm ARMA}}}
\newcommand{\AR}{{\mbox{\rm AR}}}
\newcommand{\MA}{{\mbox{\rm MA}}}
\newcommand{\CPV}{{\mbox{\rm CPV}}}

\newcommand{\wh}{\widehat}
\newcommand{\wt}{\widetilde}
\newcommand{\halmos}{\quad\hfill\mbox{$\Box$}}  

\newcommand{\bfx}{\mathbf{X}}

\allowdisplaybreaks
\begin{document}


\maketitle

\begin{abstract}
{This work is devoted to functional ARMA$(p,q)$ processes and approximating vector models based on functional PCA in the context of prediction. 
After deriving sufficient conditions for the existence of a stationary solution to both the functional and the vector model equations, the structure of the approximating vector model is investigated. 
The stationary vector process is used to predict the functional process. A bound for the difference between vector and functional best linear predictor is derived. 
The paper concludes by applying functional ARMA processes for the modeling and prediction of highway traffic data.}
\end{abstract}


\section{Introduction}\label{s1} 

A {\em macroscopic highway traffic model} involves velocity, flow (number of vehicles passing a reference point per unit of time), and density (number of vehicles on a given
road segment). 
The relation among these three variables is depicted in diagrams of ``velocity-flow relation'' and ``flow-density relation''.
The diagram of ``flow-density relation'' is also called {\em fundamental diagram of traffic flow} and can be used to determine the capacity of a road system and give guidance for inflow regulations or speed limits.
Figures~\ref{ch6_flowspeed} and~\ref{ch6_flowdensity} depict these quantities for traffic data provided by the Autobahndirektion S\"udbayern. 
At a {critical traffic density} (65~veh/km) the state of flow will change from stable to unstable. 

\begin{figure}[t]
\begin{center}
\includegraphics[trim=0.4cm 0.4cm 0.4cm 0.9cm, clip=true,scale=0.72]{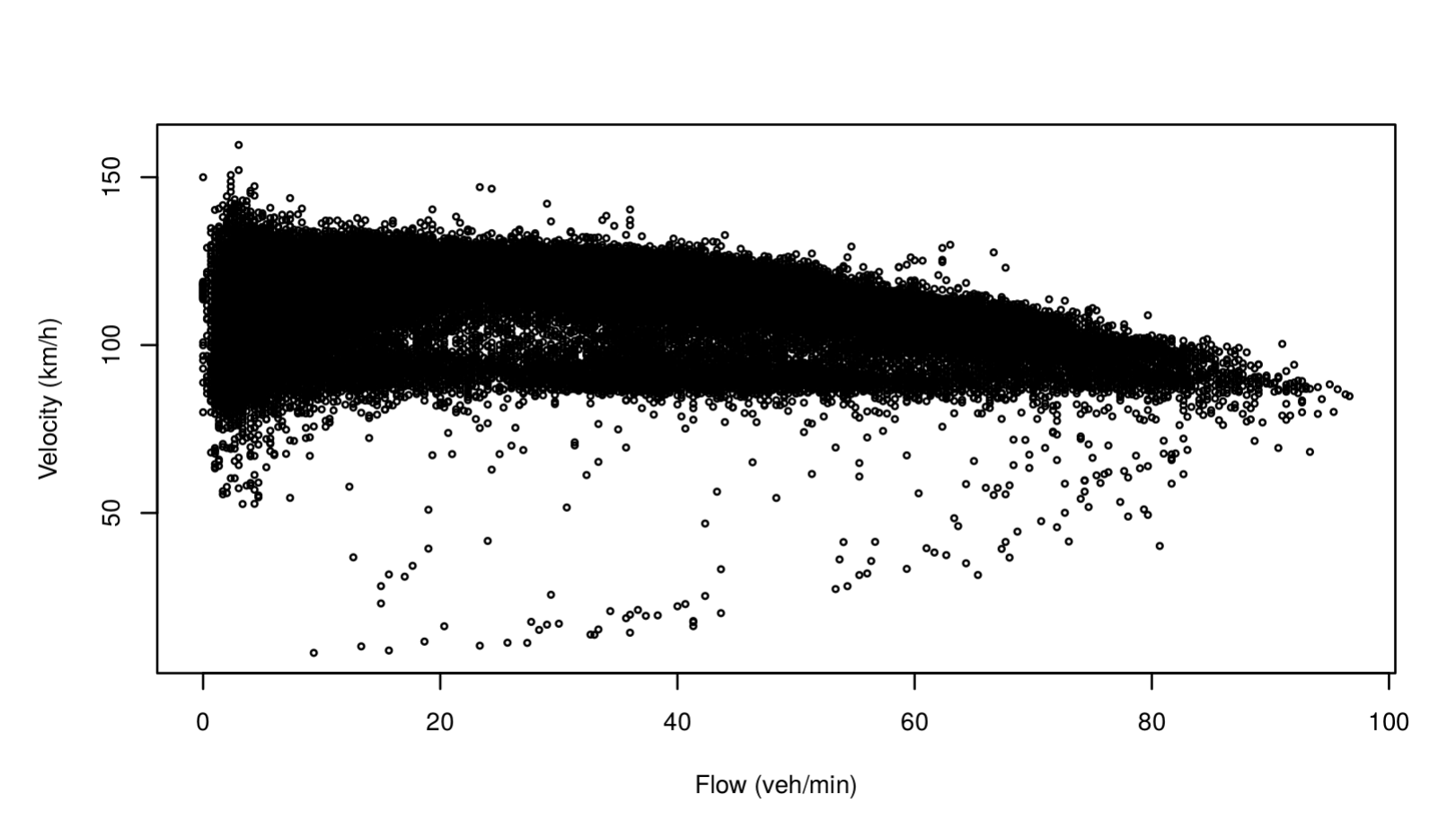}
  \caption{Velocity-flow relation on highway A92 in Southern Bavaria. 
  Depicted are average velocities per 3 min versus number of vehicles within these 3 min during the period 01/01/2014 0:00 to 30/06/2014 23:59.}
  \label{ch6_flowspeed}
  \end{center}
\end{figure}

\begin{figure}[t]
\begin{center}
\includegraphics[trim=0.4cm 0cm 0.4cm 0.7cm, clip=true,scale=0.72]{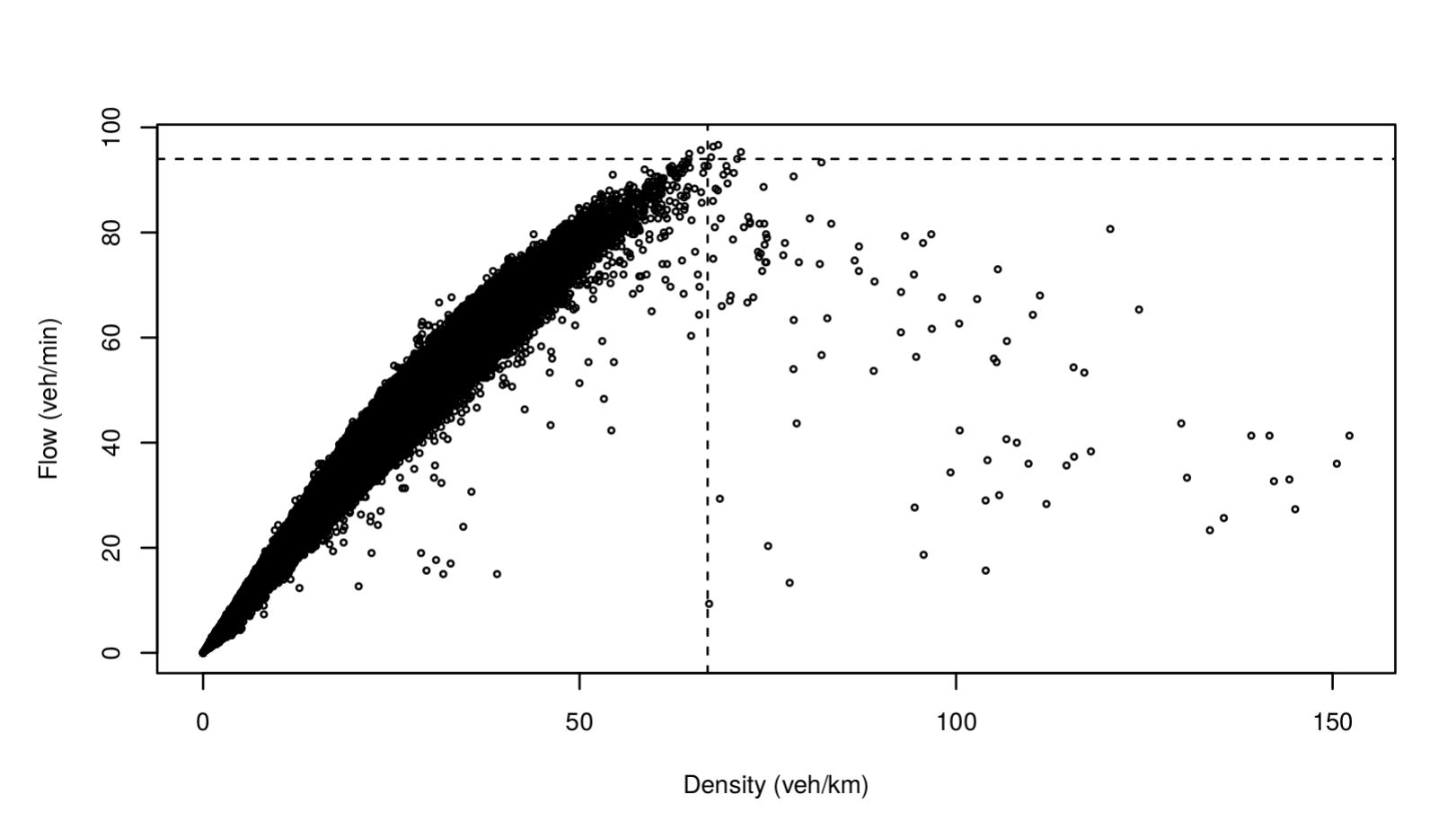}
  \caption{ Flow-density relation for the data from Figure~\ref{ch6_flowspeed} with critical traffic density of 65~veh/km.}
  \label{ch6_flowdensity}
  \end{center}
\end{figure}

In this paper we develop a {\em statistical highway traffic model} and apply it to the above data.
As can be seen from Figures~\ref{weekday_mean} and \ref{Meansubstraction_indiv} the data show a certain pattern over the day, which we want to capture utilising tools from functional data analysis. 
Functional data analysis is applied to represent the very high-dimensional traffic velocity data over the day by a random function $X(\cdot)$. 
This is a standard procedure, and we refer to  \citet{ramsay1} for details.

Given the functional data, we want to assess temporal dependence between different days; i.e., our goal is a realistic time series model for functional data, which captures the day to day dependence. 
Our analysis can support short term traffic regulation realised in real-time by electronic devices during the day, which may benefit from a more precise and parsimonious day-to-day prediction.

From a statistical point of view we are interested in the prediction of a functional ARMA$(p,q)$ process for arbitrary orders $p$ and $q$.
In scalar and multivariate time series analysis there exist several prediction methods, which can be easily implemented like the Durbin-Levinson and the Innovations Algorithm (see e.g \citet{brockwell}).
For functional time series, \citet{bosq} has proposed the \emph{functional best linear predictor} for a general linear process. 
However, implementation of the predictor is in general not feasible, because explicit formulas of the predictor can not be derived.  
{The class of functional  AR$(p)$ processes is an exception, where explicit prediction formulas have been given  (e.g. \citet{bosq}, Chapter~3, and \citet{kargin}).}
The functional AR$(1)$ model has also been applied to the prediction of traffic data in \citet{bessecardot}.

In \citet{aue} a prediction algorithm is proposed, which combines the idea of functional principal component analysis (FPCA) and functional time series analysis.
The basic idea is to reduce the infinite-dimensional functional data by FPCA to vector data. 
Thus, the task of predicting a functional time series is transformed to the prediction of a multivariate time series. 
In \cite{aue} this algorithm is used to predict 
the functional AR$(1)$ process. 

In this paper we focus on functional ARMA$(p,q)$ processes. 
{We start by providing sufficient conditions for the existence of a stationary solution to functional ARMA$(p,q)$ models. 
Then we obtain a vector process by projecting the functional process on the linear span of the $d$ most important eigenfunctions of the covariance operator of the process.  
We derive conditions such that the projected process follows a vector ARMA$(p,q)$. 
If these conditions do not hold, we show that the projected process can at least be approximated by a vector ARMA$(p,q)$ process, and we assess the quality of the approximation.
We present conditions such that the vector model equation has a unique stationary solution.}
This leads to prediction methods for functional ARMA$(p,q)$ processes.
An extension of the prediction algorithm of \citet{aue} can be applied, and makes sense under stationarity of both the functional and the vector ARMA$(p,q)$ process. 
We derive bounds for the difference between vector and functional best linear predictor. 

An extended simulation study can be found in \citet{taoran}, Chapter~5, and confirms that approximating the projection of a functional ARMA process by a vector ARMA process of the same order works reasonably well.

Our paper is organised as follows. In Section~\ref{s2} we introduce the necessary Hilbert space theory and notation, which we use throughout. We present the Karhunen-Lo\`eve Theorem and describe the FPCA based on the functional covariance operator. 
In Section~\ref{s4} we turn to functional time series models with special emphasis on functional ARMA$(p,q)$ processes. 
Section~\ref{s42} is devoted to stationarity conditions for the functional ARMA$(p,q)$ model.
In Section~\ref{s43} we study the vector process obtained by projection of the functional process on the linear span of the $d$ most important eigenfunctions of the covariance operator. 
We investigate its stationarity and prove that a vector ARMA process approximates the functional ARMA process in a natural way.
Section~\ref{s5} investigates the prediction algorithm for functional ARMA$(p,q)$ processes invoking the vector process, and compares it to the functional best linear predictor.
Finally, in Section~\ref{s6} we apply  our results to traffic data of velocity measurements. 

\section{Methodology}\label{s2}

We summarize some concepts which we shall use throughout. 
For details and more background we refer to the monographs \citet{bosq}, \citet{horvath} and  \citet{hsing}.
Let $H=L^2\left(\left[0,1\right]\right)$ be the real separable Hilbert space 
of square integrable functions $x: \left[0,1\right]\to\R$ with norm $\|x\|=(\int_0^1 x^2(s) ds)^{1/2}$ generated by  the inner product
\begin{equation*}
\left\langle x,y\right\rangle:=\int_0^1 x(t)y(t)dt,\quad x,y\in L^2\left(\left[0,1\right]\right). 
\end{equation*}
 We shall often use Parseval's equality, which ensures that for an orthonormal basis (ONB) $(e_i)_{i\in\N}$
\begin{align} \label{parseval}
\langle x,y\rangle = \sum_{i=1}^\infty \langle x,e_i\rangle \langle e_i, y\rangle,\quad x,y\in H.
\end{align}
We denote by $\mathcal{L}$ the space of bounded linear operators acting on $H$. 
If not stated differently, we take the standard operator norm defined for a
bounded operator $\Psi\in\mathcal{L}$ by $\|\Psi\|_{\mathcal{L}}:=\sup_{\|x\|\leq 1}\|\Psi x\|$.
 
A bounded linear operator $\Psi$ is a \emph{Hilbert-Schmidt} operator if it is compact and for every ONB $(e_i)_{i\in\N}$ of $H$
\begin{displaymath}
\sum_{i=1}^{\infty}\|\Psi e_i\|^2<\infty.
\end{displaymath}
We denote by $\mathcal{S}$ the space of Hilbert-Schmidt operators acting on $H$, which is again a separable Hilbert space equipped with the following inner product and corresponding \emph{Hilbert-Schmidt norm}:
\begin{align*}
&\left\langle \Psi_1,\Psi_2\right\rangle_\mathcal{S}:=\sum\limits_{i=1}^\infty\left\langle \Psi_1 e_i,\Psi_2 e_i\right\rangle\quad\mbox{and}\quad
\|\Psi\|_\mathcal{S}:=\sqrt{\left\langle\Psi,\Psi\right\rangle_\mathcal{S}}=\sqrt{\sum_{i=1}^{\infty}\|\Psi  e_i\|^2}<\infty.
\end{align*}
If $\Psi$ is a Hilbert-Schmidt operator, then 
\begin{align*}
\|\Psi\|_{\mathcal{L}}\leq\|\Psi\|_{\mathcal{S}}.
\end{align*}

Let $\mathcal{B}_H$ be the Borel $\sigma$-algebra of subsets of $H$. 
All random functions are defined on some probability space $\left(\Omega,\mathcal{A},P\right)$
and are $\mathcal{A}-\mathcal{B}_H$-measurable.
Then the space of square integrable  random functions $L^2_H:=L^2_H(\Omega,\mathcal{A},P)$ is a Hilbert space with inner product $E\left\langle X,Y\right\rangle = E \int_0^1 X(s) Y(s) ds$ for $X,Y\in L^2_H$.
We call such $X$ an \emph{$H$-valued random function}. 
For $X\in L_H^2$ there is a unique function $\mu \in H$, the  \emph{functional mean} of $X$, such that $E\langle y,X\rangle = \langle y,\mu \rangle $ for $y\in H$, satisfying 
\begin{equation*}
\mu(t)=E[X(t)],\quad t\in[0,1].
\end{equation*}
We assume throughout that $\mu = 0$, since under weak assumptions on $X$ the functional mean can be estimated consistently  from the data (see Remark~\ref{consistent}).

\bde
The \emph{covariance operator} $C_X$ of $X$ acts on $H$ and is defined as
\begin{align}\label{ch2_covoperatorform}
C _X:  x\mapsto E\left[\langle X,x\rangle X \right],\quad x\in H.
\end{align}
More precisely,
\begin{align*}
(C_X x)(t) &=E\left[\int_{0}^{1} X(s) x(s)\mathrm{d}s\, X(t)\right]
=\int_{0}^{1}E\left[X(t) X(s)\right] x(s) ds .
\end{align*}
\ede

$C_X$ is a symmetric, non-negative definite Hilbert-Schmidt operator with spectral representation
$$C_X x = \sum_{j=1}^\infty \la_j\langle x,\nu_j\rangle \nu_j,\quad x\in H,$$
for eigenpairs $(\la_j, \nu_j)_{j\in\N}$, where $(\nu_j)_{j\in\N}$ is an ONB of $H$ and $(\la_j)_{j\in\N}$ is a sequence of positive real numbers such that $\sum_{j=1}^\infty\la_j<\infty$.
When considering spectral representations we assume that the $\la_j$ are decreasingly ordered; i.e., $\la_i\ge \la_k$ for $i<k$.
Every $X\in L_H^2$ can be represented as a linear combination of the eigenfunctions $(\nu_i)_{i\in\N}$. This is known as the \emph{Karhunen-Lo\`eve representation}.

\bthe[Karhunen-Lo\`eve Theorem]\label{theorem2.3}
For $X\in L^2_H$ with $EX=0$ 
\begin{equation}\label{ch2_karhunen}
X=\sum\limits_{i=1}^{\infty}\left\langle X,\nu_i\right\rangle\nu_i,
\end{equation}
where $(\nu_i)_{i\in\N}$ are the  eigenfunctions of the covariance operator $C_X$.
The scalar products $\langle X, \nu_i\rangle$ have mean-zero, variance $\lambda_i$ and are uncorrelated; i.e., for all $i,j\in\N$,  $i\neq j$,
\begin{align}\label{ch2_lamdaandvariance}
E\left\langle X,\nu_i\right\rangle=0,\quad
E[\left\langle X,\nu_i\right\rangle \left\langle X,\nu_j\right\rangle]=0,\quad\mbox{and}\quad
E\left\langle X,\nu_i\right\rangle^2=\lambda_i,
\end{align}
where $(\lambda_i)_{i\in\N}$ are the eigenvalues of $C_X$.
\ethe

The scalar products $(\langle X, \nu_i\rangle)_{i\in\N}$ defined in (\ref{ch2_karhunen}) are called the \emph{scores} of $X$. 
By the last equation in (\ref{ch2_lamdaandvariance}), we have
\begin{equation}\label{ch2_sumoflamda}
\sum_{j=1}^{\infty}\lambda_j=\sum_{j=1}^{\infty}E\left\langle X,\nu_j\right\rangle^2=E\|X\|^2<\infty,\quad X\in L^2_H.
\end{equation}
Combining (\ref{ch2_lamdaandvariance}) and (\ref{ch2_sumoflamda}),  every $\lambda_j$ represents some proportion of the total  variability of $X$. 

\brem[The CVP method]\label{CVP}
For $d\in\N$ consider the largest $d$ eigenvalues $\lambda_1, \ldots, \lambda_d$  of $C_X$. 
The  \emph{cumulative percentage of total variance} \CPV$(d)$ is defined as
\begin{equation*}
\CPV(d):=\sum\limits_{j=1}^{d}\lambda_j\,  \big/ \, \sum\limits_{j=1}^{\infty}\lambda_j.
\end{equation*}
If we choose $d\in\N$ such that the $\CPV(d)$
exceeds a predetermined high percentage value, then  $\lambda_1,\dots,\lambda_d$ 
 explain most of the variability of $X$. 
In this context $\nu_1,\dots,\nu_d$ are called the \emph{functional principal components} (FPCs). 
\erem
\section{Functional  ARMA processes}\label{s4}

In this section we introduce the functional ARMA$(p,q)$ equations and derive sufficient conditions for the equations to have a stationary and causal solution, which we present explicitly.
We then project the functional linear process on a finite dimensional subspace of $H$. We approximate this finite dimensional process by a suitable vector ARMA process, and give conditions for the stationarity of this vector process. We also give conditions on the functional ARMA model such that the projection of the functional process on a finite dimensional space exactly follows a vector ARMA structure.

We start by defining functional white noise.

\bde[\citet{bosq}, Definition 3.1]\label{def_wn}\\
Let $(\eps_n)_{n\in\Z}$ be a sequence of $H$-valued random functions.\\
(i) \, $(\eps_n)_{n\in\Z}$  is \emph{$H$-white noise} (WN) if
for all $n\in\mathbb{Z}$, $E[\varepsilon_n]=0$, $0<E\|\varepsilon_n\|^2=\sigma^2_\varepsilon<\infty$, $C_{\varepsilon_n}=C_{\varepsilon}$, and if $C_{\varepsilon_n,\varepsilon_m}(\cdot) := E[ \left\langle \varepsilon_m,\cdot \right\rangle  \varepsilon_n ]=0$ for all $n\neq m$.\\
(ii) $(\eps_n)_{n\in\Z}$ is  \emph{$H$-strong white noise} (SWN),
if for all $n\in\mathbb{Z}$, $E[\varepsilon_n]=0$, $0<E\|\varepsilon_n\|^2=\sigma^2_\varepsilon<\infty$ and $(\eps_n)_{n\in\Z}$ is i.i.d.
\ede

We assume throughout that $(\eps_n)_{n\in\Z}$ is WN  with zero mean and $E\|\eps_n\|^2=\sigma^2_\eps<\infty$. When SWN is required, this will be specified.

\subsection{Stationary functional \ARMA\ processes}\label{s42}

Formally we can define a functional ARMA process of arbitrary order.

\bde\label{armadef}
Let $(\eps_n)_{n\in\Z}$ be WN as in Definition~\ref{def_wn}(i). 
Let furthermore $\phi_1,\dots,\phi_p$, $\theta_1,\ldots,\theta_q\in\mathcal{L}$. 
Then a solution of
\begin{equation}\label{ch4_FARMA}
X_n=\sum_{i=1}^{p}\phi_{i}X_{n-i}+\sum_{j=1}^q\theta_j\varepsilon_{n-j} + \eps_{n},\quad n\in\Z,
\end{equation} 
is called a {\em functional} ARMA$(p,q)$ {\em process}.
\ede

We  derive conditions such that \eqref{ch4_FARMA} has a stationary solution. 
We begin with the functional ARMA$(1,q)$ process and need the following assumption.
\begin{assumption}\label{farma}
There exists some $j_0\in\N$ such that $\|\phi_1^{j_0}\|_{\mathcal{L}}<1$.
\end{assumption}

\bthe\label{theorem4.6}
Let $(X_n)_{n\in\Z}$ be as in Definition~\ref{armadef} with $p=1$ and set $\phi_1=:\phi$.
If Assumption~\ref{farma} holds, there exists a unique stationary and
causal solution to (\ref{ch4_FARMA}) given by
\begin{align}
X_n &= \eps_n + (\phi + \theta_1)\eps_{n-1} + (\phi^2+\phi\theta_1+\theta_2)\eps_{n-2}\nonumber \\
& \quad +\cdots + (\phi^{q-1}+\phi^{q-2}\theta_1+\cdots + \theta_{q-1})\eps_{n-(q-1)} \notag\\
& \quad + \sum_{j=q}^\infty \phi^{j-q}(\phi^q+\phi^{q-1}\theta_1+\cdots + \theta_{q})\eps_{n-j}\nonumber\\
& = \; \sum_{j=0}^{q-1} (\sum_{k=0}^j \phi^{j-k} \theta_k ) \eps_{n-j} + \sum_{j=q}^\infty \phi^{j-q} ( \sum_{k=0}^q \phi^{q-k}\theta_k ) \eps_{n-j}\label{solution},
\end{align}
where $\phi^0=I$ denotes the identity operator in $H$. 
Furthermore, the series in (\ref{solution}) converges in $L^2_H$ and with probability one. 
\ethe

 For the proof we need the following lemma.
 
\ble[\citet{bosq}, Lemma 3.1]\label{lemma4.4}
For every $\phi\in\mathcal{L}$ the following are equivalent: \\
(i) \, There exists some $j_0\in\N$ such that $\|\phi^{j_0}\|_{\mathcal{L}}<1$.\\
(ii) \, There exist $a > 0$ and $0 < b < 1$ such that $\|\phi^j\|_{\mathcal{L}}<ab^j$ for every $j\in\N$.
\ele

\noindent
{\em Proof of Theorem~\ref{theorem4.6}.  }
We follow the lines of the proof of Proposition~3.1.1 of \citet{brockwell} and Theorem~3.1 in \cite{bosq}. First we prove $L^2_H$-convergence of the series (\ref{solution}). 
  Take $m\ge q$ and consider the truncated series 
 \begin{align}
 X_n^{(m)} &:= \eps_n + (\phi + \theta_1)\eps_{n-1} + (\phi^2+\phi\theta_1+\theta_2)\eps_{n-2}\nonumber \\
& \quad +\cdots + (\phi^{q-1}+\phi^{q-2}\theta_1+\cdots + \theta_{q-1})\eps_{n-(q-1)} \notag\\
& \quad+ \sum_{j=q}^m \phi^{j-q}(\phi^q+\phi^{q-1}\theta_1+\cdots + \theta_{q})\eps_{n-j}. \label{truncatedarma}
 \end{align}
 Define 
$$\beta(\phi,\theta) := \phi^q+\phi^{q-1}\theta_1+\cdots +\phi\theta_{q-1}+ \theta_{q}\in\call.$$
Since $(\varepsilon_n)_{n\in\Z}$ is WN, for all $m'>m\geq q$, 
\begin{align*}
E\|X^{(m')}_n&-X^{(m)}_n\|^2
=E\Big\|\sum_{j=m}^{m'}\phi^{j-q}\beta(\phi,\theta)\varepsilon_{n-j}\Big\|^2\\
&=\sum_{j=m}^{m'}E\left\|\phi^{j-q}\beta(\phi,\theta)\varepsilon_{n-j}\right\|^2\\
&\leq \sigma^2_\varepsilon \sum_{j=m}^{m'}\left\|\phi^{j-q}\right\|^2_\mathcal{L} \|\beta(\phi,\theta)\|^2_\mathcal{L}.
\end{align*}
 Lemma~\ref{lemma4.4} applies giving
\begin{equation}\label{ch4_operatorbound}
\sum_{j=0}^{\infty}\|\phi^j\|^2_{\mathcal{L}}<\sum_{j=0}^{\infty}a^2b^{2j}=\frac{a^2}{1-b^2}<\infty.
\end{equation}
Thus,
\begin{equation*}
\sum_{j=m}^{m'}\left\|\phi^{j-q}\right\|^2_\mathcal{L} \|\beta(\phi,\theta)\|^2_\mathcal{L}
\leq \|\beta(\phi,\theta)\|^2_\mathcal{L} \, a^2 \sum_{j=m}^{m'}b^{2(j-q)}\to 0,\quad\text{as}\ m,m'\to\infty.
\end{equation*}
By the Cauchy criterion the series in (\ref{solution}) converges in $L^2_H$.\\[2mm]
To prove convergence with probability one we investigate the following second moment, using that $(\varepsilon_n)_{n\in\Z}$ is WN:
\begin{align*}
E\Big(\sum_{j=1}^{\infty}\big\Vert\phi^{j-q}\beta(\phi,\theta)&\varepsilon_{n-j}\big\Vert\Big)^2
\leq E\Big(\sum_{j=1}^{\infty}\Vert\phi^{j-q}\Vert_{\call} \Vert\beta(\phi,\theta)\Vert_{\call} \Vert\varepsilon_{n-j}\Vert\Big)^2\\
&\leq \sigma^2_\varepsilon\Vert\beta(\phi,\theta)\Vert_{\call}^2  \big(\sum_{j=1}^{\infty}\Vert\phi^{j-q}\Vert_\mathcal{L}\big)^2 .
\end{align*}
Finiteness follows, since by $(\ref{ch4_operatorbound})$, 
\begin{align*}
\left\|\beta(\phi,\theta)\right\|^2_\mathcal{L}
\Big(\sum_{j=1}^{\infty} \|\phi^{j-q}\|^2_\mathcal{L}\Big)^2
& <\left\|\beta(\phi,\theta)\right\|^2_\mathcal{L} \Big(\sum_{j=1}^{\infty}ab^{j-q}\Big)^2
=\left\|\beta(\phi,\theta)\right\|^2_\mathcal{L}\frac{a^2}{(1-b)^2}<\infty.
\end{align*}
Thus, the series \eqref{solution} converges with probability one.

Note that the solution  (\ref{solution}) is stationary, since its second order structure only depends on $(\varepsilon_n)_{n\in\Z}$, which is as WN shift invariant.\\[2mm]
In order to prove that \eqref{solution} is a solution of (\ref{ch4_FARMA}) with $p=1$, we plug (\ref{solution}) into (\ref{ch4_FARMA}), and obtain for $n\in\Z$,
\begin{align}
X_n-\phi X_{n-1} &=\; \sum_{j=0}^{q-1} (\sum_{k=0}^j \phi^{j-k} \theta_k ) \eps_{n-j} + \sum_{j=q}^\infty \phi^{j-q} ( \sum_{k=0}^q \phi^{q-k}\theta_k ) \eps_{n-j} \notag \\
 & \quad - \phi\Big(\; \sum_{j=0}^{q-1} (\sum_{k=0}^j \phi^{j-k} \theta_k ) \eps_{n-1-j} + \sum_{j=q}^\infty \phi^{j-q} ( \sum_{k=0}^q \phi^{q-k}\theta_k ) \eps_{n-1-j}\Big) \label{plugin}.
\end{align}
The third term of the right-hand side 
 can be written as
\beao
&&\sum_{j=0}^{q-1} (\sum_{k=0}^j \phi^{j+1-k} \theta_k ) \eps_{n-1-j} + \sum_{j=q}^\infty \phi^{j+1-q} ( \sum_{k=0}^q \phi^{q-k}\theta_k ) \eps_{n-1-j} \\
 &=& \sum_{j'=1}^{q} (\sum_{k=0}^{j'-1} \phi^{j'-k} \theta_k ) \eps_{n-j'} + \sum_{j'=q+1}^\infty \phi^{j'-q} ( \sum_{k=0}^q \phi^{q-k}\theta_k ) \eps_{n-j'}\\
 &=& \sum_{j'=1}^{q} (\sum_{k=0}^{j'} \phi^{j'-k} \theta_k - \phi^{j'-j'}\theta_{j'} ) \eps_{n-j'} + \sum_{j'=q+1}^\infty \phi^{j'-q} ( \sum_{k=0}^q \phi^{q-k}\theta_k ) \eps_{n-j'}\\
 &=& \sum_{j'=1}^{q} (\sum_{k=0}^{j'} \phi^{j'-k} \theta_k ) \eps_{n-j'} + \sum_{j'=q+1}^\infty \phi^{j'-q} ( \sum_{k=0}^q \phi^{q-k}\theta_k ) \eps_{n-j'} - \sum_{j'=1}^q \theta_{j'} \eps_{n-j'} .
 \eeao
Comparing the sums in \eqref{plugin}, the only remaining terms are
\begin{align*}
X_n-\phi X_{n-1} & = \eps_n - \sum_{k=0}^{q} \phi^{q-k} \theta_k  \eps_{n-q} + \sum_{j'=1}^q \theta_{j'} \eps_{n-j'} + \sum_{k=0}^{q} \phi^{q-k} \theta_k  \eps_{n-q}\\
 &= \eps_n + \sum_{j'=1}^q \theta_{j'} \eps_{n-j'},\quad n\in\Z,
\end{align*}
which shows that (\ref{solution}) is a solution of equation (\ref{ch4_FARMA}) with $p=1$.\\[2mm]
 Finally, we prove uniqueness of the solution. 
 Assume that there is another stationary solution $X'_n$ of \eqref{ch4_FARMA}. Iteration gives (cf. \cite{spangenberg}, eq. (4)) for all $r>q$,
\begin{align*}
X'_n &= \sum_{j=0}^{q-1} (\sum_{k=0}^j \phi^{j-k}\theta_k) \eps_{n-j} 
+ \sum_{j=q}^{r-1} \phi^{j-q} (\sum_{k=0}^q \phi^{q-k} \theta_k)\eps_{n-j}\\
& \quad\quad + \sum_{j=0}^{q-1} \phi^{r+j-q}  (\sum_{k=j+1}^q \phi^{q-k} \theta_k) \eps_{n-(r+j)} + \phi^r X'_{n-r} .
\end{align*}
Therefore, with $X^{(r)}$ as in \eqref{truncatedarma}, for $r>q$,
\begin{align*}
E\big\|X'_n- & X^{(r)}_n\big\|^2 =E\Big\|
 \sum_{j=0}^{q-1} \phi^{r+j-q}  \, (\sum_{k=j+1}^q \phi^{q-k} \theta_k) \eps_{n-(r+j)} + \phi^r X'_{n-r}  \Big\|^2\\
&\leq 2 E\Big\| \sum_{j=0}^{q-1} \phi^{r+j-q}  \, (\sum_{k=j+1}^q \phi^{q-k} \theta_k) \eps_{n-(r+j)}   \Big\|^2
+2 \ E\left\|\phi^{r}X'_{n-r}\right\|^2\\
&\leq 2 \Vert \phi^{r-q} \Vert_{\call} ^2  E\Big\| \sum_{j=0}^{q-1} \phi^{j}  \, (\sum_{k=j+1}^q \phi^{q-k} \theta_k) \eps_{n-(r+j)}   \Big\|^2
+2 \Vert\phi^r\Vert_{\call}^2 E\left\|X'_{n-r}\right\|^2.
\end{align*}
Since both $(\varepsilon_n)_{n\in\Z}$ and  $(X'_n)_{n\in\Z}$ are stationary, Lemma~\ref{lemma4.4} yields
\begin{align*}
E\big\|X'_n- X^{(r)}_n\big\|^2 \rightarrow 0,\quad r\rightarrow \infty.
\end{align*}
Thus $X'_n$ is in $L^2_H$ equal to the limit $X_n$ of $X_n^{(r)}$, which proves  uniqueness. 
\halmos

\brem
In \citet{spangenberg}  a strictly stationary, not necessarily causal solution of a functional ARMA$(p,q)$ equation in Banach spaces is derived under minimal conditions. This extends known results considerably.
\erem

For a functional ARMA$(p,q)$ process we use the state space representation 
\begin{equation}\label{ch4_statespaceequation}
\underbrace{\begin{pmatrix}
  X_n \\
  X_{n-1}\\
  \vdots\\
  X_{n-p+1}
 \end{pmatrix}}_{\mathlarger {Y_n}}=
 \underbrace{\begin{pmatrix}
 \phi_1 &   \cdots & \phi_{p-1} & \phi_p\\
   I    &       &       &0\\
       &  \ddots &   &  \vdots     \\
      &  & I & 0
  \end{pmatrix}}_{\mathlarger {\wt{\phi}}}
 \underbrace{\begin{pmatrix}
    X_{n-1} \\
    X_{n-2}\\
    \vdots\\
    X_{n-p}
   \end{pmatrix}}_{\mathlarger Y_{n-1}}
   +\sum_{j=0}^q 
    \underbrace{\begin{pmatrix}
 \theta_j &   0 & \hdots & 0\\
   0    &   0    &       &\vdots\\
    \vdots   &   &\ddots   &    \\
     0 &  &  & 0
  \end{pmatrix}}_{\mathlarger {\wt\theta_j}}
  \underbrace{\begin{pmatrix}
       \varepsilon_{n-j} \\
       0\\
       \vdots\\
       0
      \end{pmatrix},}_{\mathlarger {\delta_{n-j}}}
\end{equation} 
where $\theta_0=I$, and $I$ and $0$ denote the identity and zero operators on $H$, respectively. 
We summarize this as
\begin{equation}\label{statespace}
Y_n=\wt{\phi}\,Y_{n-1}+\sum_{j=0}^q\wt{\theta}_j \delta_{n-j},\quad n\in\Z.
\end{equation}
Since $X_n$ and $\eps_n$ take values in $H$, 
$Y_n$ and $\delta_n$ take values in the product Hilbert space {$H^p:=(L^2([0,1]))^p$ with inner product and norm given by
\begin{equation}\label{4.13}
\left\langle x,y\right\rangle_p:=\sum_{j=1}^{p}\left\langle x_j,y_j \right\rangle
\quad\mbox{and}\quad
\|x\|_p:=\sqrt{\langle x,x\rangle_p}.
\end{equation}
We denote by $\call(H^p)$ the space of bounded linear operators acting on $H^p$. 
The operator norm of $\wt{\phi}\in\call(H^p)$ is defined as usual by
$\|\wt{\phi}\|_{\mathcal{L}}:=\sup_{\|x\|_p\leq 1} \|\wt{\phi}\,x\|_p$.
The random vector $(\delta_n)_{n\in\mathbb{Z}}$ is WN in $H^p$.}
 Let $P_1$ be the projection of $H^p$ on the first component; i.e.,
\begin{align*}
P_1 x=x_1,\quad x=(x_1,\ldots,x_n)^\top \in H^p. 
\end{align*}
\begin{assumption}\label{farmap}
There exists some $j_0 \in \N$ such that $\wt{\phi}$ as in \eqref{ch4_statespaceequation} satisfies $\|\wt{\phi}^{j_0}\|_{\mathcal{L}}<1$.
\end{assumption}

Since the proof of Theorem~\ref{theorem4.6} holds also in $H^p$, using the state space representation of a functional ARMA$(p,q)$ in $H$ as a functional ARMA$(1,q)$ in $H^p$, we get the following theorem as a consequence of Theorem~\ref{theorem4.6}.

\bthe \label{theo_statespace}
Under Assumption~\ref{farmap} there exists a unique stationary and causal solution to the functional \ARMA$(p,q)$ equations (\ref{ch4_FARMA}).
The solution can be written as $X_n=P_1 Y_n$, where $Y_n$ is the solution to the state space equation \eqref{statespace}, given by
\begin{align*}
Y_n &= \delta_n + (\wt{\phi} + \wt{\theta}_1)\delta_{n-1} + ({\wt{\phi}}^{\, 2}+\wt{\phi}\;\wt{\theta}_1+\wt{\theta}_2)\delta_{n-2}\nonumber \\
& \ +\cdots + (\wt{\phi}^{\, q-1}+\wt{\phi}^{\, q-2}\,\wt{\theta}_1+\cdots + \wt{\theta}_{q-1})\delta_{n-(q-1)}\nonumber \\
& \  + \sum_{j=q}^\infty \wt{\phi}^{\, j-q}(\wt{\phi}^{\, q}+\wt{\phi}^{\, q-1}\,\wt{\theta}_1+\cdots + \wt{\theta}_q)\delta_{n-j},\nonumber\\
& = \; \sum_{j=0}^{q-1} (\sum_{k=0}^j \wt{\phi}^{j-k} \wt{\theta}_k ) \delta_{n-j} + \sum_{j=q}^\infty \wt{\phi}^{j-q} ( \sum_{k=0}^q \wt{\phi}^{q-k}\,\wt{\theta}_k ) \delta_{n-j},
\end{align*}
where $\wt{\phi}^{\, 0}$ denotes the identity operator in $H^p$ and
$Y_n$, $\delta_n$, $\wt{\phi}$ and $\wt{\theta}_1$,\dots,$\wt{\theta}_q$ are defined in \eqref{ch4_statespaceequation}.
Furthermore, the series converges in $L^2_H$ and with probability one.
\ethe

\subsection{The  vector \ARMA$(p,q)$ process}\label{s43}

We project the stationary functional ARMA$(p,q)$ process $(X_n)_{n\in\Z}$ on a finite-dimen\-sional subspace of $H$.
{We fix $d\in\N$ 
and consider the projection of $(X_n)_{n\in\Z}$ on the subspace $\sp\{\nu_1,\dots,\nu_d\}$  spanned by the $d$ most important eigenfunctions of $C_X$} giving
\begin{align}\label{Xtrunc}
X_{n,d}=P_{\sp\{\nu_1,\dots,\nu_d\}}X_n=\sum_{i=1}^d \langle X_n, \nu_i\rangle \nu_i.
\end{align}

\brem \label{timedep}{
The dimension reduction based on the principal components is optimal for uncorrelated data in terms of its $L^2$-accuracy (cf. \citet{horvath}, Section~3.2). 
We consider time series data, where dimensions corresponding to eigenfunctions $\nu_l$ for $l>d$ can have an impact on subsequent elements of the time series, even if the corresponding eigenvalue $\lambda_l$ is small. Hence FPCA might not be optimal for  functional time series.

In \citet{hoermann} and \citet{panaretros2} an optimal dimension reduction for dependent data is introduced. They propose a  filtering technique based on a frequency domain approach, which reduces the dimension in such a way that the score vectors form a multivariate time series with diagonal lagged covariance  matrices. 
However, as pointed out in \citet{aue}, it is unclear how the technique can be utilized for prediction, since both future and past observations are required. 


In order not to miss information valuable for prediction when reducing the dimension, we include cross validation on the prediction errors to choose the number of FPCs used to represent the data (see Section~5).  
This also allows us to derive explicit bounds for the prediction error in terms of the eigenvalues of $C_X$ (see Section~4).}
\erem 
	
In what follows we are interested in 
\begin{align}\label{4.23}
\mathbf{X}_n&:=\left(\left\langle X_n,\nu_1\right\rangle,\dots,\left\langle X_n,\nu_d\right\rangle \right)^\top.
\end{align}
$\mathbf{X}_n$ is $d$-dimensional and  {isometrically isomorph} to $X_{n,d}$ (e.g. \cite{hsing}, Theorem~2.4.17).

\brem \label{consistent} 
For theoretical considerations of the prediction problem we assume that $C_X$ and its eigenfunctions are known. 
In a statistical data analysis the eigenfunctions have to replaced by their empirical counterparts.
In order to ensure consistency of the estimators we need slightly stronger assumptions on the innovation process $(\varepsilon_n)_{n\in\Z}$  and on the model parameters, similarly as for estimation and prediction in classical time series models (see \citet{brockwell}).

In \citet{weaklydep} it is shown that, under \mbox{$L^4-m$} approximability  (a weak dependence concept for functional processes), empirical estimators of mean and covariance of the functional process are $\sqrt{n}$-consistent. 
Estimated eigenfunctions and eigenvalues inherit $\sqrt{n}$-consistency results from the estimated covariance operator (Theorem~3.2 in \cite{weaklydep}). 
Proposition~2.1 of \cite{weaklydep} states conditions on the parameters of a linear process to ensure that the time series is \mbox{$L^4-m$} approximable, which are satisfied for stationary functional ARMA processes, where the WN has a finite 4-th moment.
\erem

Our next result, which follows from the linearity of the projection operator, concerns the projection of the WN $(\varepsilon_n)_{n\in\Z}$ on $\sp\{\nu_1,\dots,\nu_d\}$.

\ble\label{ch2_whitenoisetheorem}
Let $(e_i)_{i\in\N}$ be an arbitrary ONB of $H$. 
For $d\in\N$ we define the $d$-dimensional vector process  
\begin{equation*}
\mathbf{Z}_n:=(\left\langle \varepsilon_n,e_1\right\rangle,\dots,\left\langle \varepsilon_n,e_d\right\rangle)^\top,\quad n\in\mathbb{Z}.
\end{equation*} 
(i) \,  If $(\eps_n)_{n\in\Z}$ is {\rm WN} as in Definition~\ref{def_wn}(i),
then $(\mathbf{Z}_n)_{n\in\Z}$ is  {\rm WN} in $\R^d$.\\
(ii) \, If $(\eps_n)_{n\in\Z}$ is {\rm SWN} as in Definition~\ref{def_wn}(ii),
then $(\mathbf{Z}_n)_{n\in\Z}$ is {\rm SWN} in $\R^d$.
\ele

As in Section~\ref{s42} we start with the functional ARMA$(1,q)$ process for $q\in\N$ and are interested in the dynamics of $(X_{n,d})_{n\in\Z}$ of \eqref{Xtrunc} for fixed $d\in\N$. 
Using the model equation \eqref{ch4_FARMA} with $p=1$ and $\phi_1=\phi$, we get 
\begin{equation}\label{ch4_expansion}
\left\langle X_{n},\nu_l\right\rangle=
\left\langle\phi X_{n-1},\nu_l\right\rangle + \sum_{j=0}^q\left\langle\theta_j\varepsilon_{n-j},\nu_l\right\rangle,\quad l\in\mathbb{Z}.
\end{equation}
For every $l$ we expand $\left\langle\phi X_{n-1},\nu_l\right\rangle$, using that $(\nu_l)_{l\in\N}$ is an ONB of $H$ as
\begin{align*}
\left\langle\phi X_{n-1},\nu_l\right\rangle
&=\Big\langle\phi\Big(\sum_{l'=1}^{\infty}\left\langle X_{n-1},\nu_{l'}\right\rangle \nu_{l'}\Big),\nu_l\Big\rangle
=\sum_{l'=1}^{\infty}\left\langle\phi\nu_{l'},\nu_l\right\rangle\left\langle X_{n-1},\nu_{l'}\right\rangle,
\end{align*}
and $\left\langle\theta_j\varepsilon_{n-j},\nu_l\right\rangle$ for $j=1,\ldots,q$ as
\begin{align*}
\left\langle\theta_j\varepsilon_{n-j},\nu_l\right\rangle
&=\Big\langle\theta_j\Big(\sum_{l'=1}^{\infty}\left\langle \varepsilon_{n-j},\nu_{l'}\right\rangle \nu_{l'}\Big),\nu_l\Big\rangle
=\sum_{l'=1}^{\infty}\left\langle\theta_j\nu_{l'},\nu_l\right\rangle\left\langle \varepsilon_{n-j},\nu_{l'}\right\rangle.
\end{align*} 
In order to study the $d$-dimensional vector process $(\bfx_n)_{n\in\Z}$, for notational ease, we restrict a precise presentation to the \ARMA$(1,1)$ model.
The presentation of the \ARMA$(1,q)$ model is an obvious extension.

For a matrix representation of $\mathbf{X}_n$ given in \eqref{4.23} consider the notation: 
\begin{align*}
\left(\begin{array}{c|c}
\mathbf{\Phi} & \mathbf{\Phi}^\infty \\
\hline
\vdots\ &\vdots
\end{array}\right) = {\scriptstyle
\left(\begin{array}{ccc|cc}
\left\langle\phi\nu_1,\nu_1\right\rangle  &\dots &\left\langle\phi\nu_d,\nu_1\right\rangle &\left\langle\phi\nu_{d+1},\nu_1\right\rangle & \dots \\
\vdots & \ddots &\vdots &\vdots &\ddots\\
\left\langle\phi\nu_1,\nu_d\right\rangle  &\dots &\left\langle\phi\nu_d,\nu_d\right\rangle &\left\langle\phi\nu_{d+1},\nu_d\right\rangle &\dots\\
\hline
\left\langle\phi\nu_1,\nu_{d+1}\right\rangle  &\dots &\left\langle\nu_d,\nu_{d+1}\right\rangle &\left\langle\phi\nu_{d+1},\nu_{d+1}\right\rangle &\dots\\
\vdots & \ddots &\vdots &\vdots &\ddots
\end{array}
\right).}
\end{align*}
The matrices $\mathbf{\Theta}$ and $\mathbf{\Theta}^\infty$ are defined analogously. For $q=1$, with $\theta_0=I$ and $\theta_1=\theta$, (\ref{ch4_expansion}) is given in matrix form  by
\begin{equation}\label{ch4_simplifiedblockmatrix}
\begin{pmatrix}
\mathbf{X}_{n}\\
\hline
\mathbf{X}_{n}^\infty
\end{pmatrix}
=\left[
\begin{array}{c|c}
\mathbf{\Phi} & \mathbf{\Phi}^\infty \\
\hline
\vdots\ &\vdots
\end{array}
\right]\begin{pmatrix}
\mathbf{X}_{n-1}\\
\hline
\mathbf{X}_{n-1}^\infty
\end{pmatrix}+\begin{pmatrix}
\mathbf{E}_{n}\\
\hline
\mathbf{E}_{n}^\infty
\end{pmatrix}+\left[
\begin{array}{c|c}
\mathbf{\Theta} & \mathbf{\Theta}^\infty \\
\hline
\vdots\ &\vdots
\end{array}
\right]\begin{pmatrix}
\mathbf{E}_{n-1}\\
\hline
\mathbf{E}_{n-1}^\infty
\end{pmatrix},
\end{equation}
where 
\begin{align}
\mathbf{E}_n&:=\left(\left\langle \varepsilon_n,\nu_1\right\rangle,\dots,\left\langle \varepsilon_n,\nu_d\right\rangle \right)^\top,\nonumber\\
\mathbf{X}_n^\infty&:=\left(\left\langle X_n,\nu_{d+1}\right\rangle,\dots\right)^\top, \; \text{and}\nonumber\\
\mathbf{E}_n^\infty&:=\left(\left\langle \varepsilon_n,\nu_{d+1}\right\rangle,\dots\right)^\top.\nonumber
\end{align}
The operators $\mathbf{\Phi}$ and $\mathbf{\Theta}$ in \eqref{ch4_simplifiedblockmatrix} are $d\times d$ matrices with entries
 $\left\langle\phi\nu_{l'},\nu_l\right\rangle$ and $\left\langle\theta\nu_{l'},\nu_l\right\rangle$ in the  $l$-th row and $l'$-th column, respectively. 
 Furthermore,
 $\mathbf{\Phi}^\infty$ and $\mathbf{\Theta}^\infty$ are $d\times\infty$ matrices with $ll'$-th entries $\left\langle\phi\nu_{l'+d},\nu_l\right\rangle$ and $\left\langle\theta\nu_{l'+d},\nu_l\right\rangle$, respectively.
  
By (\ref{ch4_simplifiedblockmatrix}) $(\mathbf{X}_n)_{n\in\Z}$ satisfies the $d$-dimensional vector equation
\begin{align}\label{ch4_vectorprocess}
\mathbf{X}_{n}
&=\mathbf{\Phi}\mathbf{X}_{n-1}+\mathbf{E}_{n}+\mathbf{\Theta}\mathbf{E}_{n-1}+\mathbf{\Delta}_{n-1},\quad n\in\Z,
\end{align}
where 
\beam\label{delta}
\mathbf{\Delta}_{n-1}:=\mathbf{\Phi}^\infty\mathbf{X}^\infty_{n-1} +
\mathbf{\Theta}^\infty\mathbf{E}^\infty_{n-1}.
\eeam
By Lemma~\ref{ch2_whitenoisetheorem} $\left(\mathbf{E}_n\right)_{n\in\mathbb{Z}}$ is $d$-dimensional WN.
Note that $\mathbf{\Delta}_{n-1}$ in (\ref{delta}) is a $d$-dimensional vector with $l$-th component 
\begin{align}\label{ch4_decoposedeltan}
\left(\mathbf{\Delta}_{n-1}\right)_l 
=\sum_{l'=d+1}^\infty\left\langle\phi\nu_{l'},\nu_l\right\rangle\left\langle X_{n-1},\nu_{l'}\right\rangle +\sum_{l'=d+1}^\infty\left\langle\theta\nu_{l'},\nu_l \right\rangle\left\langle\varepsilon_{n-1},\nu_{l'}\right\rangle.
\end{align}
Thus, the ``error term'' $\mathbf{\Delta}_{n-1}$ depends on $X_{n-1}$, and the vector process $(\mathbf{X}_n)_{n\in\mathbb{Z}}$ in (\ref{ch4_vectorprocess}) is in general {not} a vector \ARMA$(1,1)$ process with innovations $(\mathbf{E}_{n})_{n\in \Z}$. 
 However, we can use a vector \ARMA\ model as an approximation to $(\mathbf{X}_n)_{n\in\Z}$, where we can make $\mathbf{\Delta}_{n-1}$ arbitrarily small by increasing the dimension $d$. 

\ble\label{lemma5.11}
Let $\|\cdot\|_2$ denote the Euclidean norm in $\R^d$, and let the $d$-dimensional vector $\mathbf{\Delta}_{n-1}$ be defined as in (\ref{delta}). 
Then $E\|\mathbf{\Delta}_{n-1}\|_2^2$ is bounded and tends to 0 as $d\to\infty$.
\ele

\bproof
From \eqref{delta} we obtain 
\begin{align}\label{ch4_bound2}
E\|\mathbf{\Delta}_{n-1}\|_2^2
&\leq 2\left(E\|\mathbf{\Phi}^\infty\mathbf{X}^\infty_{n-1}\|_2^2+E\|\mathbf{\Theta}^\infty\mathbf{E}^\infty_{n-1}\|_2^2\right).
\end{align}
We estimate the two parts $E\|\mathbf{\Phi}^\infty\mathbf{X}^\infty_{n-1}\|_2^2$ and $E\|\mathbf{\Theta}^\infty\mathbf{E}^\infty_{n-1}\|_2^2$ separately.
 By (\ref{ch4_decoposedeltan}) we obtain (applying Parseval's equality \eqref{parseval} in the third line),
\begin{align}
E\|\mathbf{\Phi}^\infty\mathbf{X}^\infty_{n-1}\|_2^2
& =E\Big[\sum_{l=1}^{d}\Big(\sum_{l'=d+1}^\infty\left\langle\langle X_{n-1},\nu_{l'}\rangle\phi\nu_{l'},\nu_l\right\rangle\Big)^2\Big]\notag\\
&\leq E\Big[\sum_{l=1}^{\infty}\Big\langle\sum_{l'=d+1}^\infty \langle X_{n-1},\nu_{l'}\rangle\phi\nu_{l'},\nu_l\Big\rangle^2\Big]\notag\\
&=E\Big\|\sum_{l'=d+1}^\infty \langle X_{n-1},\nu_{l'}\rangle\phi\nu_{l'}\Big\|^2.\notag
\end{align}
Since  the scores $(\langle X_{n-1,l},\nu_l\rangle)_{l\in\N}$ are uncorrelated  (cf. the Karhunen-Lo\`eve Theorem~\ref{theorem2.3}), and then using monotone convergence, we find
\begin{align*}
E\|\mathbf{\Phi}^\infty\mathbf{X}^\infty_{n-1}\|_2^2 \leq E \sum_{l'=d+1}^{\infty}\langle X_{n-1},\nu_{l'}\rangle^2  \Vert\phi\nu_{l'}\Vert^2
=\sum_{l'=d+1}^{\infty}E\left(\langle X_{n-1},\nu_{l'}\rangle\right)^2\left\|\phi\nu_{l'}\right\|^2.
\end{align*}
Since by \eqref{ch2_lamdaandvariance} $E\langle X_{n-1},\nu_{l'}\rangle^2=\lambda_{l'}$, we get
\begin{align}\label{ch4_bound1}
\sum_{l'=d+1}^{\infty}E\left(\langle X_{n-1},\nu_{l'}\rangle\right)^2\left\|\phi\nu_{l'}\right\|^2&=\sum_{l'=d+1}^{\infty}\lambda_{l'}\Vert\phi\Vert_{\call}^2\Vert\nu_{l'}\Vert^2
\leq \Vert\phi\Vert_{\call}^2\sum_{l'=d+1}^{\infty}\lambda_{l'}.
\end{align}
The bound for $E\|\mathbf{\Theta}^\infty\mathbf{E}^\infty_{n-1}\|_2^2$ can be obtained in exactly the same way, and we calculate
\begin{align}\label{ch4_boundthetaepsilon}
E\|\mathbf{\Theta}^\infty\mathbf{E}^\infty_{n-1}\|_2^2\notag
&\leq\sum_{l'=d+1}^{\infty}E\langle\varepsilon_{n-1},\nu_{l'}\rangle^2\left\|\theta\nu_{l'}\right\|^2\\
&\leq \Vert \theta\Vert^2_{\call}\sum_{l'=d+1}^{\infty} E\langle \langle \varepsilon_{n-1},\nu_{l'}\rangle\varepsilon_{n-1},\nu_{l'}\rangle \notag\\
&=\Vert \theta\Vert^2_{\call}\sum_{l'=d+1}^{\infty} \langle C_{\varepsilon}\nu_{l'},\nu_{l'}\rangle,
\end{align}
where $C_{\varepsilon}$ is the covariance operator of the WN. 
As a covariance operator it has finite nuclear operator norm
$\Vert C_{\varepsilon}\Vert_{\mathcal{N}}:=\sum_{l'=1}^{\infty} \langle C_{\varepsilon}(\nu_{l'}),\nu_{l'}\rangle<\infty$.
 Hence,  $\sum_{l'=d+1}^{\infty} \langle C_{\varepsilon}\nu_{l'},\nu_{l'}\rangle\rightarrow 0 $ for $d\rightarrow \infty$.
Combining (\ref{ch4_bound2}), (\ref{ch4_bound1}) and (\ref{ch4_boundthetaepsilon}) we find that $E\|\mathbf{\Delta}_{n-1}\|_2^2$ is bounded and tends to 0 as $d\to\infty$.
\eproof

For the vector ARMA$(1,q)$ model the proof of boundedness of $E\|\mathbf{\Delta}_{n-1}\|_2^2$ is analogous.
We now summarize our findings for a functional ARMA$(1,q)$ process. 
 
\bthe\label{ch4_stationaryvector}
Consider a functional \ARMA$(1,q)$ process for $q\in\N$ such that Assumption~\ref{farma} holds. 
For $d\in\N$, the vector process of \eqref{4.23} 
has the representation
\beao
\bfx_n = \mathbf{\Phi}\mathbf{X}_{n-1}+\mathbf{E}_{n}+\sum_{j=1}^q\mathbf{\Theta}_q\mathbf{E}_{n-j}+\mathbf{\Delta}_{n-1},\quad n\in\Z,
\eeao
where 
$$\mathbf{\Delta}_{n-1}:=\mathbf{\Phi}^\infty\bfx_{n-1}^\infty + \sum_{j=1}^q\mathbf{\Theta}_j^\infty\mathbf{E}_{n-j},$$ 
and all quantities are defined analogously to \eqref{4.23}, \eqref{ch4_vectorprocess}, and \eqref{delta}.
Define 
\begin{align}\label{vectorarma}
\check{\bfx}_n = \mathbf{\Phi}\check{\bfx}_{n-1}+\mathbf{E}_{n}+\sum_{j=1}^q\mathbf{\Theta}_j\mathbf{E}_{n-j} ,\quad  n\in\mathbb{Z}.
 \end{align}
Then  both the functional \ARMA$(1,q)$ process $(X_n)_{n\in\Z}$ in \eqref{ch4_FARMA} and the $d$-dimensional vector process $(\check\bfx_n)_{n\in\Z}$ in \eqref{vectorarma} have a unique stationary and causal solution.
Moreover, $E\|\mathbf{\Delta}_{n-1}\|_2^2$ is bounded and tends to 0 as $d\to\infty$.
\ethe

\begin{proof}
Recall from \eqref{ch4_simplifiedblockmatrix} the $d\times d$ matrix $\mathbf{\Phi}$ 
of the vector process \eqref{vectorarma}.
In order to show that \eqref{vectorarma} has a stationary solution, by Theorem 11.3.1 of \cite{brockwell},  it suffices to prove that every eigenvalue $\lambda_k$ of $\mathbf{\Phi}$ with corresponding eigenvector  $\textbf{a}_k=(\textbf{a}_{k,1},\dots,\textbf{a}_{k,d}) $ satisfies  $\vert\lambda_k\vert<1$ for $k=1,\dots,d$. 
Note that $\vert\lambda_k\vert<1$ is equivalent to $\vert \lambda_k^{j_0} \vert <1$ for all $j_0\in\N$.
Define $a_k := \textbf{a}_{k,1}\nu_1+ \dots + \textbf{a}_{k,d}\nu_d \in H$, then by Parseval's equality \eqref{parseval},
$\Vert a_k\Vert^2=\sum_{l=1}^d \vert \langle a_{k},\nu_l\rangle \vert^2=\sum_{l=1}^d \mathbf{a}_{k,l}^2=\Vert\mathbf{a}_k\Vert^2_2=1$ for $k=1,\dots,d$.
With the orthogonality of $\nu_1,\dots,\nu_d$ we find 
$	\Vert \mathbf{\Phi}\textbf{a}_{k} \Vert_2^2 = \sum_{l=1}^d \big(\sum_{l'=1}^d\langle \phi\nu_{l'},\nu_l\rangle \textbf{a}_{k,l}\big)^2$.
Defining $A_d=\sp\{\nu_1,\dots,\nu_d\}$, we calculate
\begin{align*}
	\Vert P_{A_d}\phi P_{A_d} a_k \Vert^2 
	= \sum_{l=1}^d\big\langle\phi(\sum_{l'=1}^d \textbf{a}_{k,l'}\nu_{l'}),\nu_l\big\rangle^2 \Vert \nu_l \Vert ^2= \sum_{l=1}^d\big(\sum_{l'=1}^d\textbf{a}_{k,l'}\langle\phi \nu_{l'},\nu_l\rangle\big)^2 = \Vert \Phi \textbf{a}_{k} \Vert_2^2.
\end{align*}
Hence, for $j_0$ as in Assumption~\ref{farma}, 
\begin{align*}
\vert \lambda_k^{j_0} \vert& = \Vert \lambda_k^{j_0} \textbf{a}_k \Vert_2 = \Vert \mathbf{\Phi}^{j_0} \textbf{a}_k \Vert_2 = \big\Vert \big(P_{A_d} \phi P_{A_d}\big)^{j_0} a_k\big\Vert \\& \leq \big\Vert \big(P_{A_d} \phi P_{A_d}\big)^{j_0} \big\Vert_{\call} \Vert a_k \Vert \leq \Vert \phi^{j_0} \Vert_{\call} < 1,
\end{align*}
which finishes the proof.
\end{proof}

In order to extend approximation \eqref{vectorarma} of a functional ARMA$(1,q)$ process  to a functional ARMA$(p,q)$ process we use again the state space representation (\ref{statespace}) given by
\begin{equation*}
Y_n=\wt{\phi}Y_{n-1}+\sum_{j=0}^q\wt{\theta}_j\delta_{n-j},\quad n\in\Z,
\end{equation*}
where $Y_n$, $\wt{\theta}_0=I$, $\wt{\phi}$, $\wt{\theta}_1,\dots,\wt{\theta}_q$ and $\delta_n$ are defined as in Theorem \ref{theo_statespace} and take values in $H_p=\left(L^2([0,1])\right)^p$; cf. \eqref{4.13}.

\bthe\label{cor5.14}
Consider the functional \ARMA$(p,q)$ process as defined in \eqref{ch4_FARMA} such that Assumption~\ref{farmap} holds. 
Then for $d\in\N$ the vector process of \eqref{4.23}
has the  representation
\beam\label{vectorprocess_pq}
\bfx_n =\sum_{i=1}^p \mathbf{\Phi}_i\mathbf{X}_{n-i}+ \mathbf{E}_{n}+\sum_{j=1}^q\mathbf{\Theta}_q\mathbf{E}_{n-j} +\mathbf{\Delta}_{n-1},\quad n\in\Z,
\eeam
where 
$$\mathbf{\Delta}_{n-1}:=\sum_{i=1}^p\mathbf{\Phi}_i^\infty\bfx_{n-i}^\infty + \sum_{j=1}^q\mathbf{\Theta}_j^\infty\mathbf{E}_{n-j},$$ 
and all quantities are defined analogously to \eqref{4.23}, \eqref{ch4_vectorprocess}, and \eqref{delta}.
Define 
\beam\label{vectorarma_pq}
 \check\bfx_n = \sum_{i=1}^p\mathbf{\Phi}_i\check\bfx_{n-i}+\mathbf{E}_{n}+\sum_{j=1}^q\mathbf{\Theta}_q\mathbf{E}_{n-1} ,\quad  n\in\mathbb{Z}.
 \eeam
Then  both the functional \ARMA$(p,q)$ process $(X_n)_{n\in\Z}$ in \eqref{ch4_FARMA} and the $d$-dimensional vector process $(\check\bfx_n)_{n\in\Z}$ in \eqref{vectorarma_pq} have a unique stationary and causal solution.
Moreover, $E\|\mathbf{\Delta}_{n-1}\|_2^2$ is bounded and tends to 0 as $d\to\infty$.
\ethe

We are now interested in conditions for $(\mathbf{X}_n)_{n\in\Z}$ to exactly follow a vector ARMA$(p,q)$ model.
A trivial condition is that the projections of $\phi_i$ and $\theta_j$ on $A_d^\perp$, the orthogonal complement  of $A_d=\sp\{\nu_1,\ldots,\nu_d\}$, satisfy
\beao
P_{A_d^{\perp}}\phi_i P_{A_d^{\perp}} = P_{A_d^{\perp}}\theta_jP_{A_d^{\perp}}= 0
\eeao
 for all $i=1,\dots,p$ and $j=1,\dots,q$.
In that case $\check\bfx_n = \mathbf{X}_n$ for all $n\in\Z$.

However, as we show next, the assumptions on the moving average parameters $\theta_1,\dots,\theta_q$ are actually not required. We start with a well-known result that characterizes vector MA processes.

\begin{lemma}[\citet{brockwell}, Proposition~3.2.1] \label{maq}
If $(\mathbf{X}_n)_{n\in\mathbb{Z}}$ is a stationary vector process with autocovariance matrix $\mathbf{C}_{\mathbf{X}_h,\mathbf{X}_0}=E[\mathbf{X}_h \mathbf{X}^\top_0]$ with $\mathbf{C}_{\mathbf{X}_q,\mathbf{X}_0} \neq 0$ and $\mathbf{C}_{\mathbf{X}_h,\mathbf{X}_0}=0$ for $|h|>q$, then $(\mathbf{X}_n)_{n\in\mathbb{Z}}$ is a vector \MA$(q)$.
\end{lemma}

\begin{proposition}
Let $A_d=\sp\{\nu_1,\ldots,\nu_d\}$ and $A_d^{\perp}$ its orthogonal complement.
If  $P_{A_d^{\perp}}\phi_iP_{A_d^{\perp}}=0$ for all $i=1,\dots,p$,
then the $d$-dimensional process $(\bfx_n)_{n\in\Z}$ as in (\ref{vectorprocess_pq}) is a vector \ARMA$(p,q)$ process.
\end{proposition}

\bproof 
Since $\phi_i$ for $i=1,\dots,p$ only acts on $A_d$, from \eqref{vectorprocess_pq} we get
\begin{align*}
\mathbf{X}_{n}&=\sum_{i=1}^{p}\mathbf{\Phi}_i\mathbf{X}_{n-i}+\mathbf{E}_{n}+\sum_{j=1}^q\mathbf{\Theta}_j\mathbf{E}_{n-j}+\mathbf{\Delta}_{n-1}\\
&=\sum_{i=1}^{p}\mathbf{\Phi}_i\mathbf{X}_{n-i}+\mathbf{E}_n+\sum_{j=1}^q\mathbf{\Theta}_j\mathbf{E}_{n-j}+ \sum_{j=1}^q\mathbf{\Theta}_j^\infty\mathbf{E}^\infty_{n-j},\quad n\in\Z.
\end{align*}
To ensure that $(\mathbf{X}_n)_{n\in\Z}$ follows a vector ARMA$(p,q)$ process, we have to show that $$\mathbf{R}_n:=\mathbf{E}_{n}+\sum_{j=1}^q\mathbf{\Theta}_j\mathbf{E}_{n-j}+ \sum_{j=1}^q\mathbf{\Theta}_j^\infty\mathbf{E}^\infty_{n-j},\quad n\in\Z,$$ 
follows a vector MA$(q)$ model. 
According to Lemma~\ref{maq} it is sufficient to verify that $(\mathbf{R}_n)_{n\in\Z}$ is stationary and has an appropriate autocovariance structure. 

Defining (with $\theta_0=I$)
\begin{align*}
R_n:=\sum_{j=0}^q\theta_j\varepsilon_{n-j},\quad n\in\Z,
\end{align*}
where $\theta_1,\dots,\theta_q$ are as in \eqref{ch4_FARMA}, observe that $\mathbf{R}_n=(\langle R_n,\nu_1\rangle,\dots, \langle R_n,\nu_d\rangle)$ is isometrically isomorph to $P_{A_d}R_n=\sum_{j=1}^d \langle R_n,\nu_j\rangle \nu_j$ for all $n\in\Z$.
Hence, stationarity of $(\mathbf{R}_n)_{n\in\Z}$ immediately follows from the stationarity of $(R_n)_{n\in\Z}$.
Furthermore, 
\begin{align*}
E[\langle P_{A_d}R_0,\cdot \rangle P_{A_d}R_h ] = P_{A_d} E [\langle R_0,\cdot\rangle R_h ] P_{A_d}= P_{A_d}C_{R_h,R_0}P_{A_d}.
\end{align*}
But since $(R_n)_{n\in\Z}$ is a functional  \MA$(q)$ process, $C_{R_h,R_0}=0$ for $\vert h \vert > q$. 
By the relation between $P_{A_d}R_n$ and $\mathbf{R}_n$ we also have  $\mathbf{C}_{\mathbf{R}_h,\mathbf{R}_0}=0$ for $|h|>q$ and, hence, $(\mathbf{R}_n)_{n\in\mathbb{Z}}$ is a vector MA$(q)$.
\eproof

\section{Prediction of functional \ARMA\ processes}\label{s5}

For $h\in\N$ we derive the best $h$-step linear predictor of a functional ARMA$(p,q)$ process $(X_n)_{n\in\Z}$ based on $\bfx_1,\dots,\bfx_n$ as defined in \eqref{vectorprocess_pq}. 
We then compare the vector best linear predictor to the functional best linear predictor based on $X_1,\dots,X_n$ and show that, under regularity conditions, the difference is bounded and tends to $0$ as $d$ tends to infinity.

\subsection{Prediction based on the vector process}\label{s51}

In finite dimensions the concept of a {best linear predictor} is well-studied. 
For a $d$-dimensional stationary time series $(\bfx_n)_{n\in\Z}$ we denote the {\em matrix linear span} of $\mathbf{X}_1,\dots,\mathbf{X}_n$ by
\beao
\mathbf{M}'_1:=\Big\{\sum_{i=1}^{n}\mathbf{A}_{ni}\mathbf{X}_i :\ \mathbf{A}_{ni}\ \text{are real} \ d\times d\ \text{matrices}, i=1,\ldots,n\Big\}. 
\eeao
Then for $h\in\N$ the \emph{$h$-step vector best linear predictor} $\wh{\mathbf{X}}_{n+h}$ of $\mathbf{X}_{n+h}$ based on $\mathbf{X}_1,\dots,\mathbf{X}_n$ is defined as the  projection of $\mathbf{X}_{n+h}$ on the closure $\mathbf{M}_1$ of $\mathbf{M}'_1$ in $L^2_{\R^d}$; i.e.,
\begin{align}\label{vectorBLP}
\wh\bfx_{n+h}:=P_{\mathbf{M}_1}\mathbf{X}_{n+h}.
\end{align}
Its properties are given by the projection theorem (e.g. Theorem~2.3.1 of \citet{brockwell}) and can be summarized as follows.

\brem\label{ch4_setofM1}
Recall that $\|\cdot\|_2$ denotes the Euclidean norm in $\R^d$ and $\langle \,,\,\rangle_{\R^d}$ the corresponding scalar product.\\
(i) \,  $E\langle\mathbf{X}_{n+h}-\wh{\mathbf{X}}_{n+h},\mathbf{Y}\rangle_{\R^d}=\mathbf{0}$ for all $\mathbf{Y}\in\mathbf{M}_1$.\\
(ii) \, $\wh{\mathbf{X}}_{n+h}$ is the unique element in $\mathbf{M}_1$ such that
$$E\|\mathbf{X}_{n+h}-\wh{\mathbf{X}}_{n+h}\|^2_2=\inf_{\mathbf{Y}\in\mathbf{M}_1}E\|\mathbf{X}_{n+h}-\mathbf{Y}\|^2_2.$$
(iii) \, $\mathbf{M}_1$ is a linear subspace of $\mathbb{R}^d$.
\erem
In analogy to the prediction algorithm suggested in \citet{aue}, a method for finding the best linear predictor of $X_{n+h}$ based on $\bfx_1,\dots,\bfx_n$ is the following: \vspace*{1pt}

\noindent\textbf{Algorithm 1} \footnote{Steps (1) and (3) are implemented in the {\tt R} package  {\tt{fda}}, and (2) in the {\tt {R}} package {\tt{mts}}}
\begin{enumerate}
\item[(1)] 
Fix $d\in\N$. Compute the FPC scores $\left<X_k,\nu_l\right>$ for $l=1,\ldots,d$ and $k=1,\ldots,n$ by projecting each $X_k$ on $\nu_1,\dots,\nu_d$. 
Summarize the scores in the vector
\begin{equation*}
\mathbf{X}_k:=(\left<X_k,\nu_1\right>,\dots,\left<X_k,\nu_d\right>)^\top,\quad k=1,\dots\,n.
\end{equation*} 
\item[(2)] 
Consider the $d$-dimensional vectors $\mathbf{X}_1,\dots,\mathbf{X}_n$. 
For $h\in\N$ compute the vector best linear predictor of $\mathbf{X}_{n+h}$ by means of \eqref{vectorBLP}:
\begin{displaymath}
\wh{\mathbf{X}}_{n+h}=(\wh{\left<X_{n+h},\nu_1\right>},\dots,\wh{\left<X_{n+h},\nu_d\right>})^\top.
\end{displaymath} 
\item[(3)] 
Re-transform the  vector best linear predictor $\wh{\mathbf{X}}_{n+h}$ into a functional form $\wh{X}_{n+h}$ via the truncated Karhunen-Lo\`eve representation: 
\begin{align}\label{KLtrunc}
\wh{X}_{n+h}&:=\wh{\left<X_{n+h},\nu_1\right>}\nu_1+\dots+\wh{\left<X_{n+h},\nu_d\right>}\nu_d.
\end{align}
\end{enumerate}

For functional \AR$(1)$ processes, \citet{aue} compare the resulting predictor \eqref{KLtrunc} to the functional best linear predictor. 
Our goal is  to extend these results to functional ARMA$(p,q)$ processes. 
However, when moving away from AR models, the best linear predictor is no longer directly given by the process. We start by recalling the notion of best linear predictors in Hilbert spaces.

\subsection{Functional best linear predictor}

For $h\in\N$ we introduce  the {\em $h$-step functional best linear predictor}  $\wh{X}_{n+h}$ of $X_{n+h}$, based on $X_1,\dots,X_n$, as proposed in \citet{bosq2014}. 
It is the projection of $X_{n+h}$ on a large enough subspace of $L^2_H$ containing $X_1,\dots,X_n$. 
More formally, we use the concept of {$\mathcal{L}$-{\em closed subspaces} as in Definition~1.1 of \citet{bosq}. 

\bde \label{lcs}
Recall that $\mathcal{L}$ denotes the space of bounded linear operators acting on $H$.
We call $G$ an $\mathcal{L}$-{\em closed subspace} (LCS) of $L^2_H$, if\\
(1) $G$ is a Hilbertian subspace of $L^2_H$.\\
(2) If $X\in G$ and $g\in\mathcal{L}$, then $g \,X\in G$.
\ede

We define 
\begin{align*}
X^{(n)}:=(X_n,\dots,X_1).
\end{align*}
By Theorem~1.8 of \cite{bosq} the LCS $G:=G_{X^{(n)}}$ generated by $X^{(n)}$
is the closure in $L^2_{H^n}$ of $G^{\prime}_{X^{(n)}}$, where
\begin{equation*}
 G^{\prime}_{X^{(n)}}:=\big\lbrace g_n \,X^{(n)} : \ g_{n}\in\mathcal{L}(H^n,H)\,\big\rbrace.
\end{equation*}
For $h\in\N$ the \emph{$h$-step functional best linear predictor} $\wh{X}^G_{n+h}$ of $X_{n+h}$ is defined as the projection of $X_{n+h}$ on $G$, which we write as
\begin{equation}\label{ch4_definitionoffunctionalpredictor}
\wh{X}^G_{n+h}:=P_{G}X_{n+h}\in G.
\end{equation}

Its properties are given by the projection theorem (e.g. Section~1.6 in \cite{bosq}) and are summarized as follows.
\brem\label{ch4_projector} 
(i) \, $E\langle X_{n+h}-\wh{X}^G_{n+h},Y\rangle=0$ for all $Y\in G.$\\
(ii)  \, $\wh{X}^G_{n+h}$ is the unique element in $G$ such that
$$E\|X_{n+h}-\wh{X}^G_{n+h}\|^2=\inf_{Y\in G}E\|X_{n+h}-Y\|^2.$$
(iii) \, The mean squared error of the functional best linear predictor $\wh{X}^G_{n+h}$ is denoted by
\begin{equation}\label{ch4_msefblp}
\sigma_{n,h}^2:=E\|X_{n+h}-\wh{X}^G_{n+h}\|^2.
\end{equation}
\erem

Since in general $G^{\prime}_{X^(n)}$ is not closed  (cf. \citet{bosq2014}, Proposition~2.1), $\wh{X}^G_{n+h}$ is not necessarily of the form $\wh{X}^G_{n+h}=g_{n}^{(h)}\, X^{(n)}$ for some $g_{n}^{(h)} \in \call(H^n,H)$.
However, the following result gives necessary and sufficient conditions for $\wh{X}^G_{n+h}$ to be represented in terms of bounded linear operators.

\begin{proposition}[Proposition~2.2, \citet{bosq2014}]\label{prop_pred_lin}
For $h\in\N$ the following are equivalent: \\
(i) \, There exists some $g \in \call(H^n,H)$ such that $C_{X^{(n)},X_{n+h}}=g\, C_{X^{(n)}}$.\\
(ii) \, $P_G X_{n+h} = g \, X^{(n)}$ for some  $g\in\call(H^n,H)$.
\end{proposition}

This result allows us to derive conditions, such that the difference between the predictors \eqref{vectorBLP} and \eqref{ch4_definitionoffunctionalpredictor} can be computed. 
Weaker conditions are needed, if $\wh{X}^G_{n+h}$ admits a representation  $\wh{X}^G_{n+h}= s_n^{(h)}  X^{(n)}$ for some Hilbert-Schmidt operator $s_{n}^{(h)}$ from $H^n$ to $H$ ($s_n^{(h)}\in \cals(H^n,H)$). 

\begin{proposition}\label{prop_pred_hil}
For $h\in\N$ the following are equivalent:\\
(i) \, There exists some $s \in \cals(H^n,H)$ such that $C_{X^{(n)},X_{n+h}}=s\, C_{X^{(n)}}$.\\
(ii) \, $P_{G}X_{n+h} = s \, X^{(n)}$ for some $s\in\cals(H^n,H)$.
\end{proposition}

\begin{proof}
The proof is similar to the proof of Proposition~\ref{prop_pred_lin}. 
Assume that (i) holds.
Then, since $C_{X^{(n)},s\,X^{(n)}}=E[\langle X^{(n)},\cdot\rangle \, s\, X^{(n)}] = s\, C_{X^{(n)}}$, we have 
$$C_{X^{(n)},X_{n+h}-s \,X^{(n)}}=0.$$
Therefore, $X_{n+h}-s\, X^{(n)} \perp X^{(n)}$ and, hence, $X_{n+h}-s\,X^{(n)} \perp G$ which gives (ii).\\
For the reverse, note that  (ii) implies
$$ C_{X^{(n)},X_{n+h}-s \, X^{(n)}}=  C_{X^{(n)},X_{n+h}-P_{G}X_{n+h}}=0.$$ 
Thus, $C_{X^{(n)},X_{n+h}}=C_{X^{(n)},s \, X^{(n)}}=s\, C_{X^{(n)}}$, which finishes the proof.
\end{proof}

We proceed with examples of processes where Proposition~\ref{prop_pred_lin} or Proposition~\ref{prop_pred_hil} apply.

\begin{example}\label{autoregressive}
Let $(X_n)_{n\in\Z}$ be a stationary  functional \AR$(p)$ process with representation
$$X_n= \varepsilon_n + \sum_{j=1}^{p} \phi_j X_{n-j}, \quad n\in\Z,$$
where $(\varepsilon_n)_{n\in\Z}$ is {\rm WN} and $\phi_j\in\cals$ are Hilbert-Schmidt operators.  
Then for $n\ge p$, Proposition~\ref{prop_pred_hil} applies for $h=1$, giving the 1-step predictor $P_{G} X_{n+1}=s_n^{(1)} \,X^{(n)}$ for some $s_n^{(1)}\in\mathcal{S}$.
\end{example}

\begin{proof}
We calculate
\begin{align*}
C_{X^{(n)},X_{n+1}}(\cdot) 
= E \big[\langle X^{(n)} , \cdot \rangle(\phi_1,\dots,\phi_p,{0},\dots,{0}) X^{(n)} \big]=  \mathbf{\phi}C_{X^{(n)}}(\cdot),
\end{align*}
where $\mathbf{\phi}=(\phi_1,\dots,\phi_p,{0},\dots,{0})\in\call(H^n,H)$. 
Now let $(e_i)_{i\in\N}$ be an ONB of $H$. 
Then $(f_j)_{j\in\N}$ with $f_1=(e_{1},0,\dots,0)^{\top}$,  $f_2=(0,e_1,0,\dots,0)^{\top}$, $\ldots$ , $f_n=(0,\dots,0,e_1)^{\top}$,  
$f_{n+1}=(e_2,0,\dots,0)^{\top}$, $f_{n+2}=(0,e_2,0,\dots,0)^{\top}$, $\ldots$ , $f_{2n} = (0, \dots, 0,e_2 )^{\top}$, $f_{2n+1}=(e_3,0,\dots,0)^{\top}, \ldots$
is an ONB of $H^n$ and, by orthogonality of  $(e_i)_{i\in\N}$, we get
\begin{align*}
 \Vert \mathbf{\phi} \Vert_{\cals}^{1/2} &= \sum_{j\in\N} \Vert\phi f_j\Vert^2
=\sum_{i\in\N} \sum_{j=1}^p \Vert \phi_j e_i \Vert ^2
=\sum_{j=1}^p \sum_{i\in\N} \Vert\phi_j e_i\Vert^2 =\sum_{j=1}^p \Vert\phi_j\Vert^2_\call  < \infty,
\end{align*}
since $\phi_j\in\cals$ for every $j=1,\dots,p$, which implies that $\phi\in \cals(H^n,H)$.
\end{proof}

\begin{example}\label{ex2}
Let $(X_n)_{n\in\Z}$ be a stationary functional \MA$(1)$ process
$$X_n=\varepsilon_n + \theta\varepsilon_{n-1},\quad  n\in\Z,$$
where $(\varepsilon_n)_{n\in\Z}$ is {\rm WN},  $\Vert\theta\Vert_{\call}<1$, $\theta\in\cals$ and $\theta$ nilpotent, such that $\Vert\theta^j\Vert_\call=0$ for $j>j_0$ for some $j_0\in\N$. 
Then for $n>j_0$, Proposition~\ref{prop_pred_hil} applies.
\end{example}

\begin{proof}
Since $\Vert\theta\Vert_{\call}<1$, $(X_n)_{n\in\Z}$ is invertible, and since $\theta$ is nilpotent, $(X_n)_{n\in\Z}$ can be represented as an AR process of finite order, where the operators in the inverse representation are still Hilbert-Schmidt operators.
Then the statement follows from the arguments of the proof of Example~\ref{autoregressive}.
\end{proof}

\begin{example}\label{ex3}
Let $(X_n)_{n\in\Z}$ be a stationary functional \MA$(1)$ process
$$X_n=\varepsilon_n + \theta \varepsilon_{n-1}, \quad n\in\Z,$$
where $(\varepsilon_n)_{n\in\Z}$ is {\rm WN}, and denote by $C_{\varepsilon}$ the covariance operator of  the WN. 
Assume that $\Vert \theta \Vert_{\call} <1$. 
If $\theta$ and $C_{\varepsilon}$ commute, Proposition~\ref{prop_pred_hil} applies.
\end{example}

\begin{proof}
Stationarity of $(X_n)_{n\in\Z}$ ensures that $C_{X_n,X_{n+1}}=C_{X_0,X_1}$.  
Let $\theta^*$ denote the adjoint operator of $\theta$. 
Since $\theta C_{\varepsilon}=C_{\varepsilon} \theta$, we have that $C_{X_1,X_0}=C_{X_0,X_1}$ which implies $\theta C_{\varepsilon} = C_{\varepsilon} \theta^*=C_{\varepsilon} \theta$. 
{Hence, $C_{\varepsilon}=C_{X_0} - \theta C_\eps \theta^* = C_{X_0} - \theta^2 C_{\varepsilon}$.  
Since $\Vert \theta\Vert_{\call} < 1$, the operator $I+\theta^2$ is invertible. 
Therefore, $C_{\varepsilon} = (I+\theta^2)^{-1} C_{X_0}$, and we get
\begin{align*} 
C_{X_1,X_0}= \theta C_{\varepsilon} =(I+\theta^2)^{-1}\theta C_{X_0}.
\end{align*}
}Furthermore, since the space $\cals$  of Hilbert-Schmidt operators forms an ideal in the space of bounded linear operators (e.g. \cite{dunford}, Theorem~VI.5.4.) and $\theta\in\cals$, also $(I+\theta^2)^{-1}\theta \in\cals$.
\end{proof}

\subsection{Bounds for the error of the vector predictor}\label{s53}

We are now ready to derive bounds for the prediction error caused by the dimension reduction. 
More precisely, for $h\in\N$ we compare the vector best linear predictor 
$\wh{X}_{n+h} =\sum_{j=1}^d \wh{\left<X_{n+h},\nu_j\right>}\nu_j 
$ 
as defined in \eqref{KLtrunc} with the functional best linear predictor
$\wh{X}^G_{n+h}=P_G X_{n+h}$ 
of \eqref{ch4_definitionoffunctionalpredictor}. We first compare them on $\sp\{\nu_1,\ldots,\nu_d\}$, where the vector representations are given by
\begin{align}\label{ch4_gprojection}
& \wh{\mathbf{X}}_{n+h}=(\wh{\left<X_{n+h},\nu_1\right>},\dots,\wh{\left<X_{n+h},\nu_d\right>})^\top, \mbox{ and} \notag \\ 
&\wh{\mathbf{X}}^G_{n+h}:=\left(\left<\wh{X}^G_{n+h},\nu_1\right>,\dots,\left<\wh{X}^G_{n+h},\nu_d\right>\right)^\top.
\end{align}

We formulate assumptions such that for $d\to\infty$ the mean squared distance between the vector best linear predictor $\wh{\mathbf{X}}_{n+h}$ and  the vector $\wh{\mathbf{X}}^G_{n+h}$ becomes arbitrarily small.

For $ l=1,\dots,d$ the $l$-th component of $ \wh{\mathbf{X}}^G_{n+h}$ is given by
\begin{align}\label{ch4_G-discompose}
\left\langle \wh{X}^G_{n+h},\nu_l\right\rangle=\Big\langle\sum_{i=1}^{n}g_{ni}^{(h)} X_i ,\nu_l\Big\rangle &=\Big\langle\sum_{i=1}^{n}\sum_{l'=1}^{\infty}\left\langle X_i,\nu_{l'}\right\rangle g_{ni}^{(h)}\nu_{l'},\nu_l\Big\rangle\notag \\ &=\sum_{i=1}^{n}\sum_{l'=1}^{\infty}\left\langle X_i,\nu_{l'}\right\rangle\left\langle g_{ni}^{(h)}\nu_{l'},\nu_l\right\rangle.
\end{align}
Using the vector representation \eqref{ch4_gprojection}, we write
\begin{small}\begin{align}\label{ch4_bigmatrix-g}
\wh{\mathbf{X}}^{G}_{n+h}&=\sum_{i=1}^{n}\left(\begin{array}{ccc|cc}
\left\langle g_{ni}^{(h)} \nu_1,\nu_1\right\rangle  &\dots &\left\langle g_{ni}^{(h)}\nu_d,\nu_1\right\rangle   &\left\langle g_{ni}^{(h)}\nu_{d+1},\nu_1\right\rangle &\dots \\
\vdots & \vdots &\vdots &\vdots &\vdots \\
\left\langle g_{ni}^{(h)}\nu_1,\nu_d\right\rangle  &\dots &\left\langle g_{ni}^{(h)}\nu_d,\nu_d\right\rangle  &\left\langle g_{ni}^{(h)}\nu_{d+1},\nu_d\right\rangle &\dots\\
\end{array}
\right)
\begin{pmatrix}
  \left\langle X_{i},\nu_1\right\rangle \\
  \vdots\\
  \left\langle X_{i},\nu_d\right\rangle  \\
  \hline
  \left\langle X_{i},\nu_{d+1}\right\rangle\\
  \vdots
 \end{pmatrix}
\nonumber \\ 
 & =: \sum_{i=1}^{n}\mathbf{G}_{ni}^{(h)}\mathbf{X}_i+\sum_{i=1}^{n}\mathbf{G}^{(h)\infty}_{ni}\mathbf{X}_i^\infty,
\end{align}
\end{small}
where
$\mathbf{G}^{(h)}_{ni}$ are  $d\times d$ matrices with $ll'$-th component 
 $\langle g_{ni}^{(h)}\nu_{l'},\nu_l\rangle$ and
 $\mathbf{G}_{ni}^\infty$ are $d\times\infty$ matrices with $ll'$-th component $\langle g_{ni}^{(h)}\nu_{d+l'},\nu_l\rangle$.
 
Moreover, for all $Y\in G$ there exist $n\in\N$ and (possibly unbounded) linear operators $t_{n1},\dots,t_{nn}$ such that
\begin{equation}\label{4.67}
Y=\sum_{i=1}^{n}t_{ni} X_i.
\end{equation}
Similarly as in \eqref{ch4_G-discompose}, we project $Y\in G$ on $\sp\{\nu_1,\dots,\nu_d\}$, which results in
\begin{align}\label{4.68}
\mathbf{Y}&:=\left(\left\langle Y,\nu_1\right\rangle,\ldots,\left\langle Y,\nu_d\right\rangle \right)^\top\nonumber\\
&=\big(\big\langle \sum_{i=1}^{n}t_{ni}X_i,\nu_1\big\rangle,\dots,\big\langle \sum_{i=1}^{n}t_{ni}X_i,\nu_d\big\rangle \big)^\top\nonumber\\
&=:\sum_{i=1}^{n}\mathbf{T}_{ni}\mathbf{X}_i+\sum_{i=1}^{n}\mathbf{T}_{ni}^\infty\mathbf{X}_i^\infty.
\end{align}
The $d\times d$ matrices $\mathbf{T}_{ni}$ and the $d\times\infty$ matrices $\mathbf{T}^\infty_{ni}$ in \eqref{4.67} are defined in the same way as $\mathbf{G}^{(h)}_{ni}$ and $\mathbf{G}_{ni}^{(h)\infty}$ in (\ref{ch4_bigmatrix-g}).
We denote by $\mathbf{M}$ the space of all $ \mathbf{Y}$:
\begin{align*}
\mathbf{M}:=\left\lbrace \mathbf{Y}=\left(\left\langle Y,\nu_1\right\rangle,\dots,\left\langle Y,\nu_d\right\rangle \right)^\top:  Y\in G\right\rbrace.
\end{align*}
Observing that for all $\mathbf{Y}_1\in\mathbf{M}_1$  there exist $d\times d$ matrices $\mathbf{A}_{n1},\dots,\mathbf{A}_{nn}$ such that $ \mathbf{Y}_1=\sum_{i=1}^{n}\mathbf{A}_{ni}\mathbf{X}_i$, there also exist operators $t_{ni}$ such that $\mathbf{T}_{ni}=\mathbf{A}_{ni}$, and $\mathbf{T}_{ni}^\infty=\mathbf{0}$, which then gives $\mathbf{Y}_1\in \mathbf{M}$. Hence $\mathbf{M}_1\subseteq\mathbf{M}$. 

Now that we have introduced the notation and the setting, we are ready to compute the mean squared distance
 $E\|\wh{\mathbf{X}}_{n+h}-\wh{\mathbf{X}}^G_{n+h}\|_2^2$.

\bthe\label{theorem4.12}
Suppose $(X_n)_{n\in\mathbb{Z}}$ is a functional \ARMA$(p,q)$ process such that Assumption~\ref{farmap} holds. 
For $h\in\N$ let $\wh{X}^G_{n+h}$ be the functional best linear predictor of $X_{n+h}$ as defined in \eqref{ch4_definitionoffunctionalpredictor} and $\wh{\mathbf{X}}^G_{n+h}$ as in (\ref{ch4_gprojection}). 
Let furthermore $\wh{\mathbf{X}}_{n+h}$ be the vector best linear predictor of $\mathbf{X}_{n+h}$ based on $\mathbf{X}_1,\dots,\mathbf{X}_n$ as in \eqref{vectorBLP}. \\
(i) In the framework of Proposition~\ref{prop_pred_lin}, and if $\sum_{l=1}^{\infty} \sqrt{\lambda_l}<\infty$,  for all $d\in\N$, 
\begin{equation*}
E\left\|\wh{\mathbf{X}}_{n+h}-\wh{\mathbf{X}}^G_{n+h}\right\|_2^2\leq 4 \,\big(\sum_{i=1}^{n} \|g_{ni}^{(h)}\|_\call \big)^{2} \, \big(\sum_{l=d+1}^{\infty} \sqrt{\lambda_l}  \big)^{2} <\infty.
\end{equation*}
(ii) In the framework of Proposition~\ref{prop_pred_hil}, for all $d\in\N$,
\begin{equation*}
E\left\|\wh{\mathbf{X}}_{n+h}-\wh{\mathbf{X}}^G_{n+h}\right\|_2^2\leq 4 \, \big(\sum_{i=1}^{n}\big(\sum_{l=d+1}^{\infty}\|g_{ni}^{(h)}\nu_l\|^2\big)^{\frac{1}{2}}\big)^2\sum_{l=d+1}^{\infty}\lambda_l<\infty.
\end{equation*}
In both cases, $E\|\wh{\mathbf{X}}_{n+h}-\wh{\mathbf{X}}^G_{n+h}\|_2^2$ tends to 0 as $d\to\infty.$ 
\ethe

We start with a technical lemma, which we need for the proof of the above Theorem.

\ble\label{strong_orthogonal}
Suppose $(X_n)_{n\in\Z}$ is a  stationary and causal functional \ARMA$(p,q)$ process and $(\nu_l)_{l\in\N}$ are the eigenfunctions of its covariance operator $C_X$.  
Then for all $ j,l\in\N$,
\begin{equation*}
E\left[\left\langle X_{n+h}-\wh{X}^G_{n+h},\nu_l\right\rangle\left\langle Y,\nu_j\right\rangle\right]=0,\quad Y\in G.
\end{equation*}
\ele

\begin{proof}
For all $j,l\in\N$ we set $s_{l,j}(\cdot):=\langle \cdot,\nu_l\rangle\nu_j.$ 
First note that for all $x\in H$  with $\Vert x \Vert \leq 1$,  
\begin{align*}
\Vert s_{l,j}x\Vert =\|\left\langle x,\nu_l\right\rangle\nu_j\|\leq \|x\|\leq 1,
\end{align*}
hence, $s_{l,j}\in\call$.  
Since $G$ is an $\call$-closed subspace, $Y\in G$ implies $s_{l,j}(Y)\in G$ and we get with Remark~\ref{ch4_projector}(i) for all $j,l\in\N$,
\begin{align*}
E \left\langle X_{n+h}-\wh{X}^G_{n+h},s_{l,j}Y\right\rangle 
=E\left[\left\langle X_{n+h}-\wh{X}^G_{n+h},\nu_l\right\rangle\left\langle Y,\nu_j\right\rangle\right]
=0.
\end{align*}
\end{proof}

\noindent
{\em Proof of Theorem~\ref{theorem4.12}.  }
First note that under both conditions $(i)$ and $(ii)$, there exist $g_{ni}^{(h)}\in\call$ such that $\wh{X}^G_{n+h}=\sum_{i=1}^n g_{ni}^{(h)} X_{n+h-i}$ and that $\cals\subset\call$.   
With the matrix representation of $\wh{\mathbf{X}}^G_{n+h}$ in \eqref{ch4_bigmatrix-g} and Lemma~\ref{strong_orthogonal} we obtain
\begin{align}\label{4.80}
\lefteqn{\sum_{j=1}^{d}E\Big[\langle Y,\nu_j\rangle\big\langle X_{n+h}-\wh{X}^G_{n+h},\nu_j\big\rangle\Big]
=E\Big\langle \mathbf{Y}, 
\mathbf{X}_{n+h}-\wh{\mathbf{X}}^G_{n+h}\Big\rangle_{\R^d}}\nonumber\\
& = E\Big\langle\mathbf{Y},\mathbf{X}_{n+h}- \sum_{i=1}^{n}\mathbf{G}_{ni}^{(h)} \mathbf{X}_i-\sum_{i=1}^{n}\mathbf{G}^{(h)\infty}_{ni}\mathbf{X}_i^\infty\Big\rangle_{\R^d}
=0,\quad  Y\in G,
\end{align}
where $\mathbf{Y}$ is defined as in \eqref{4.68}.
Since \eqref{4.80} holds for all $\mathbf{Y}\in\mathbf{M}$ and $\mathbf{M}_1\subseteq\mathbf{M}$, it especially holds for all $\mathbf{Y}_1\in\mathbf{M}_1$; i.e.,
\begin{equation}\label{ch4_set-y2=0}
E\Big\langle \mathbf{Y}_1,\mathbf{X}_{n+h }-\sum_{i=1}^{n}\mathbf{G}_{ni}^{(h)}\mathbf{X}_i -\sum_{i=1}^{n}\mathbf{G}^{(h)\infty}_{ni}\mathbf{X}_i^\infty\Big\rangle_{\R^d}=0,\quad\mathbf{Y}_1\in\mathbf{M}_1.
\end{equation}
Combining (\ref{ch4_set-y2=0}) and Remark~\ref{ch4_projector}(i), we get
\begin{equation}\label{ch4_last3}
E\Big\langle \mathbf{Y}_1,\wh{\mathbf{X}}_{n+h}-\sum_{i=1}^{n}\mathbf{G}_{ni}^{(h)}\mathbf{X}_i\Big\rangle_{\R^d} =E\Big\langle \mathbf{Y}_1,\sum_{i=1}^{n}\mathbf{G}^{(h)\infty}_{ni}\mathbf{X}_i^\infty\Big\rangle_{\R^d},\quad\mathbf{Y}_1\in\mathbf{M}_1.
\end{equation}
Since both $\wh{\mathbf{X}}_{n+h}$ and $\sum\limits_{i=1}^{n}\mathbf{G}_{ni}^{(h)}\mathbf{X}_i$ are in $\mathbf{M}_1$, (\ref{ch4_last3}) especially holds, when  
\begin{equation}\label{4.84}
\mathbf{Y}_1=\wh{\mathbf{X}}_{n+h}-\sum_{i=1}^{n}\mathbf{G}_{ni}^{(h)}\mathbf{X}_i\in\mathbf{M}.
\end{equation}
We plug $\mathbf{Y}_1$ as defined in \eqref{4.84} into \eqref{ch4_last3} and obtain
\begin{align}\label{ch4_lastinsert}
E\Big\langle \wh{\mathbf{X}}_{n+h}&-\sum\limits_{i=1}^{n}\mathbf{G}_{ni}^{(h)}\mathbf{X}_i, \wh{\mathbf{X}}_{n+h}-\sum_{i=1}^{n}\mathbf{G}_{ni}^{(h)}\mathbf{X}_i\Big\rangle_{\R^d}\notag\\ &= E\Big\langle\wh{\mathbf{X}}_{n+h}-\sum\limits_{i=1}^{n}\mathbf{G}_{ni}^{(h)}\mathbf{X}_i,\sum_{i=1}^{n}\mathbf{G}^{(h)\infty}_{ni}\mathbf{X}_i^\infty\Big\rangle_{\R^d}.
\end{align}
From the left  hand side of \eqref{ch4_lastinsert} we read off
\begin{equation}\label{ch4_last2}
E\Big\langle\wh{\mathbf{X}}_{n+h}-\sum\limits_{i=1}^{n}\mathbf{G}_{ni}^{(h)}\mathbf{X}_i, \wh{\mathbf{X}}_{n+h}-\sum_{i=1}^{n}\mathbf{G}_{ni}^{(h)}\mathbf{X}_i\Big\rangle_{\R^d} 
=E\Big\|\wh{\mathbf{X}}_{n+h}-\sum_{i=1}^{n}\mathbf{G}_{ni}^{(h)}\mathbf{X}_i\Big\|_2^2,
\end{equation}
and for the right-hand side of (\ref{ch4_lastinsert}), applying the Cauchy-Schwarz inequality twice, we get
\begin{align}\label{ch4_last1}
E\Big\langle  \wh{\mathbf{X}}_{n+h}-&\sum\limits_{i=1}^{n}\mathbf{G}_{ni}^{(h)}\mathbf{X}_i, \sum_{i=1}^{n}\mathbf{G}^{(h)\infty}_{ni}\mathbf{X}_i^\infty\Big\rangle_{\R^d}\notag\\
&\leq E\Big[\Big\|\wh{\mathbf{X}}_{n+h} -\sum\limits_{i=1}^{n} \mathbf{G}_{ni}^{(h)}\mathbf{X}_i\Big\|_2 \, \Big\|\sum_{i=1}^{n}\mathbf{G}^{(h)\infty}_{ni}\mathbf{X}_i^\infty\Big\|_2\Big]\nonumber\\
&\leq \Big(E\Big\|\wh{\mathbf{X}}_{n+h}-\sum\limits_{i=1}^{n}\mathbf{G}_{ni}^{(h)}\mathbf{X}_i\Big\|^2_2\Big)^{\frac{1}{2}} \, \Big(E\Big\|\sum_{i=1}^{n}\mathbf{G}^{(h)\infty}_{ni}\mathbf{X}_i^\infty\Big\|^2_2\Big)^{\frac{1}{2}}.
\end{align}
Dividing the right-hand side of (\ref{ch4_last2}) by the first square root on the right-hand side of (\ref{ch4_last1}) we find
\begin{equation*}
E\Big\|\wh{\mathbf{X}}_{n+h}-\sum\limits_{i=1}^{n}\mathbf{G}_{ni}^{(h)}\mathbf{X}_i\Big\|_2^2\leq E\Big\|\sum_{i=1}^{n}\mathbf{G}^{(h)\infty}_{ni}\mathbf{X}_i^\infty\Big\|^2_2.
\end{equation*}
Hence, for the mean squared distance we obtain
\begin{align*}
E\Big\|\wh{\mathbf{X}}_{n+h}-\wh{\mathbf{X}}^G_{n+h}\Big\|_2^2&=E\Big\|\wh{\mathbf{X}}_{n+h}-\sum\limits_{i=1}^{n}\mathbf{G}_{ni}^{(h)}\mathbf{X}_i- \sum_{i=1}^{n}\mathbf{G}^{(h)\infty}_{ni}\mathbf{X}_i^\infty\Big\|_2^2\nonumber\\
&\leq 2E\Big\|\wh{\mathbf{X}}_{n+h}-\sum\limits_{i=1}^{n}\mathbf{G}_{ni}^{(h)}\mathbf{X}_i\Big\|_2^2+2E\Big\|\sum_{i=1}^{n}\mathbf{G}^{(h)\infty}_{ni}\mathbf{X}_i^\infty\Big\|^2_2\nonumber\\
&\leq 4E\Big\|\sum_{i=1}^{n}\mathbf{G}^{(h)\infty}_{ni}\mathbf{X}_i^\infty\Big\|^2_2.
\end{align*}
What remains to do is to bound $\sum\limits_{i=1}^{n}\mathbf{G}_{ni}^{(h)\infty}\mathbf{X}_i^{\infty}$, which, by \eqref{ch4_bigmatrix-g}, is a $d$-dimensional vector with $l$-th component $\sum_{i=1}^{n}\sum\limits_{l'=d+1}^{\infty}\langle X_i,\nu_{l'} \rangle \langle g_{ni}^{(h)}\nu_l',\nu_l\rangle.$

$(i)$ First we consider the framework of Proposition~\ref{prop_pred_lin}. \\
We abbreviate $x_{i,l'}:=\langle X_{i},\nu_l'\rangle$ and calculate
\begin{align}\label{ch4_bigcal}
E\Big\|\sum_{i=1}^{n}\mathbf{G}_{ni}^{(h)\infty}\mathbf{X}_i^{\infty}\Big\|_2^2 
&= E\Big\|\sum_{l=1}^{d}\Big(\sum_{i=1}^{n} \sum_{l'=d+1}^{\infty}x_{i,l'}\langle g_{ni}^{(h)}\nu_l',\nu_l\rangle\Big)\nu_l\Big\|^2\nonumber\\
&\leq E\Big\|\sum_{l=1}^{\infty}\Big(\sum_{i=1}^{n}\sum_{l'=d+1}^{\infty}x_{i,l'}\langle g_{ni}^{(h)}\nu_l',\nu_l\rangle\Big)\nu_l\Big\|^2\nonumber\\
&{=}E\Big\|\sum_{i=1}^{n}\sum_{l'=d+1}^{\infty}x_{i,l'}g_{ni}^{(h)}\nu_l'\Big\|^2 
\end{align}
by Parseval's equality \eqref{parseval}. 
Then we proceed using the orthogonality of $\nu_l$ and the Cauchy-Schwarz inequality,
\begin{align}
&=E\Big[\Big\langle \sum_{i=1}^{n}\sum_{l=d+1}^{\infty}x_{i,l}g_{ni}^{(h)}\nu_l,\sum_{j=1}^{n}\sum_{l'=d+1}^{\infty}x_{j,l'}g_{nj}^{(h)}\nu_l' \Big\rangle \Big]\nonumber\\
&= \sum_{i,j=1}^{n}\sum_{l,l'=d+1}^{\infty}E(x_{i,l}x_{j,l'})\langle g_{ni}^{(h)}\nu_l,g_{nj}^{(h)}\nu_{l'}\rangle \notag\\
&\leq\Big(\sum_{i=1}^{n}\sum_{l=d+1}^{\infty}\sqrt{E(x_{i,l})^2}\|g_{ni}^{(h)}\nu_l\| \Big)^2\notag\\
&=\Big(\sum_{i=1}^{n}\sum_{l=d+1}^{\infty}\sqrt{\lambda_l}\|g_{ni}^{(h)}\nu_l\| \Big)^2,\label{last}
\end{align}
since $E\langle X_i,\nu_l \rangle ^2 = \la_l$ by \eqref{ch2_lamdaandvariance}. 
The right-hand side of \eqref{last} is bounded above by
\begin{align}
 \Big(\sum_{i=1}^{n}\sum_{l=d+1}^{\infty} \sqrt{\lambda_l} \|g_{ni}^{(h)}\|_\call \|\nu_l\|\Big) ^2\nonumber  
 = \Big(\sum_{i=1}^{n} \|g_{ni}^{(h)}\|_\call \Big)^{2} \, \Big(\sum_{l=d+1}^{\infty} \sqrt{\lambda_l}   \Big)^{2}, \nonumber
\end{align}
since $\|\nu_l\| =1$ for all $l\in\N$.
Since $g_{ni}^{(h)}\in\call$, we have $\sum_{i=1}^{n} \|g_{ni}^{(h)}\|_\call < \infty$ for all $n\in\N$ and with
 $\sum_{l=1}^{\infty} \sqrt{\lambda_l} <\infty$, the right-hand side tends to 0 as $d\to\infty$. 
 
$(ii)$ In the framework of Proposition~\ref{prop_pred_hil}  there exist $g_{ni}^{(h)}\in\cals$ such that $\wh{X}^G_{n+h}=\sum_{i=1}^n g_{ni}^{(h)} X_{n+h-i}$. 
By the Cauchy-Schwarz inequality we estimate
\begin{align}\label{4.90}
E\Big\|\sum_{i=1}^{n}\mathbf{G}_{ni}^{(h)\infty}\mathbf{X}_i^{\infty}\Big\|_2^2 &\leq \Big(\sum_{i=1}^{n}\sum_{l=d+1}^{\infty}\sqrt{\lambda_l}\|g_{ni}^{(h)}\nu_l\| \Big)^2\nonumber\\
 &\leq \Big(\sum_{i=1}^{n}\Big(\sum_{l=d+1}^{\infty}
\lambda_l\Big)^{\frac{1}{2}}\Big(\sum_{l=d+1}^{\infty}\|g_{ni}^{(h)}\nu_l\|^2\Big)^{\frac{1}{2}}\Big)^2\notag \\
&=\Big(\sum_{i=1}^{n}\Big(\sum_{l=d+1}^{\infty}\|g_{ni}^{(h)}\nu_l\|^2\Big)^{\frac{1}{2}}\Big)^2\sum_{l=d+1}^{\infty}\lambda_l.
\end{align}
Now note that $\sum_{l=d+1}^{\infty}\|g_{ni}^{(h)}\nu_l\|^2\leq \|g_{ni}^{(h)}\|_{\mathcal{S}}<\infty.$
Thus, \eqref{4.90} is bounded by
\begin{align*}
\Big(\sum_{i=1}^{n}\Big(\sum_{l=d+1}^{\infty}\|g_{ni}^{(h)}\nu_l\|^2\Big)^{1/2}\Big)^2\sum_{l=d+1}^{\infty}
\lambda_l\leq \Big(\sum_{i=1}^{n}\|g_{ni}^{(h)}\|^{1/2}_{\mathcal{S}}\Big)^2\sum_{l=d+1}^{\infty} \lambda_l<\infty,
\end{align*}
such that \eqref{4.90} tends to $0$ as $d\to\infty$.
\halmos

We are now ready to derive bounds of the mean squared prediction error 
\mbox{$E\|X_{n+h}-\wh{X}_{n+h}\|^2$}. 
 
\bthe\label{theorembound}
Consider a stationary and causal functional \ARMA$(p,q)$ process as in \eqref{ch4_FARMA}.
Then, for $h\in\N$, $\wh{X}_{n+h}$ as defined in \eqref{KLtrunc}, and $\sigma_{n,h}^2$ as defined in \eqref{ch4_msefblp}, we obtain
\begin{equation*}
E\left\|X_{n+h}-\wh{X}_{n+h}\right\|^2\leq \sigma_{n,h}^2+\gamma_{d;n;h},
\end{equation*}
where $\gamma_{d;n;h}$ can be specified as follows.\\[2mm]
(i) In the framework of Proposition~\ref{prop_pred_lin}, and if $\sum_{l=1}^{\infty} \sqrt{\lambda_l}<\infty$,  for all $d\in\N$, 
\begin{align*}
&\gamma_{d;n;h}=4 \,\Big(\sum_{i=1}^{n} \|g_{ni}^{(h)}\|_\call \Big)^{2} \, \Big(\sum_{l=d+1}^\infty \sqrt{\lambda_l}  \Big)^{2}+\sum_{l=d+1}^\infty \lambda_l.
\end{align*}
(ii) In the framework of Proposition~\ref{prop_pred_hil}, for all $d\in\N$,
\begin{align*}
&\gamma_{d;n;h}=  \sum_{l=d+1}^\infty\lambda_l \, \big(4\, g_{n;d;h} + 1\big)\quad\mbox{with}\quad g_{n;d;h}=\sum_{i=1}^n\Big(\sum_{l=d+1}^\infty\left\|g_{ni}^{(h)}\nu_l\right\|^2\Big)^{1/2}\leq \sum_{i=1}^n \|g_{ni}^{(h)}\|^2_\mathcal{S}.
\end{align*} 
In both cases, $E\left\|X_{n+h}-\wh{X}_{n+h}\right\|^2$ tends to $\sigma_{n,h}^2$ as $d\to\infty.$ 
\ethe

\bproof\label{ch4_othogonol}
{ 
With \eqref{ch2_lamdaandvariance} and since $(\nu_l)_{l\in\N}$ is an ONB, we get}
\begin{align}
E\left\|X_{n+h}-\wh{X}_{n+h}\right\|^2 &= E\Big\Vert\sum_{l=1}^d \langle X_{n+h}-\wh{X}_{n+h},\nu_l\rangle \nu_l + 
\sum_{l=d+1}^\infty\langle X_{n+h},\nu_l\rangle \nu_l \Big\Vert^2\notag\\
&=\sum_{l=1}^d  E\left\Vert\langle X_{n+h}-\wh{X}_{n+h},\nu_l\rangle \nu_l\right\Vert^2 +  \sum_{l=d+1}^\infty E\left\Vert \langle X_{n+h},\nu_l\rangle \nu_l \right\Vert^2\notag\\
&=\sum_{l=1}^d  E\langle X_{n+h}-\wh{X}_{n+h},\nu_l\rangle^2 + \sum_{l=d+1}^\infty \lambda_l. \label{proof1}
\end{align}
Now note that by definition of the Euclidean norm,
$$\sum_{l=1}^d  E\langle X_{n+h}-\wh{X}_{n+h},\nu_l\rangle^2 = E\Vert \mathbf{X}_{n+h}-\wh{\mathbf{X}}_{n+h} \Vert_2^2.$$
Furthermore, by  Definition~\ref{lcs} of $\call$-closed subspaces and Remark~\ref{ch4_projector}(i), $E\langle X_{n+h}-\wh{X}^G_{n+h},Y\rangle=0$ for all $Y\in G.$ 
Observing that $\wh{X}^G_{n+h}-\wh{X}_{n+h}\in G$, we conclude that
\begin{align*}
E\langle X_{n+h}-\wh{X}^G_{n+h},\wh{X}^G_{n+h}-\wh{X}_{n+h}\rangle=0,
\end{align*}
and, by Lemma~\ref{strong_orthogonal},
\begin{align*}
E \langle X_{n+h}-\wh{X}^G_{n+h},\nu_l\rangle
\langle\wh{X}^G_{n+h}-\wh{X}_{n+h},\nu_{l'}\rangle=0,\quad l,l'\in\N.
\end{align*}
Hence,
\begin{align}
E\Vert \mathbf{X}_{n+h}-\wh{\mathbf{X}}_{n+h} \Vert_2^2 &=E \Vert \mathbf{X}_{n+h}-\wh{\mathbf{X}}^{G}_{n+h} \Vert_2^2 + E\Vert \wh{\mathbf{X}}^{G}_{n+h}-\wh{\mathbf{X}}_{n+h} \Vert_2^2, \label{proof2}
\end{align}
where for the first term of the right-hand side, 
\begin{align}
E \Vert \mathbf{X}_{n+h}-\wh{\mathbf{X}}^{G}_{n+h} \Vert_2^2 &= E \sum_{l=1}^d \langle X_{n+h}-\wh{X}^{G}_{n+h},\nu_l\rangle^2 \leq \sum_{l=1}^{\infty} \langle X_{n+h}-\wh{X}^{G}_{n+h},\nu_l\rangle^2 \notag\\ &= E \Vert X_{n+h}-\wh{X}^G_{n+h}\Vert^2=\sigma_{n,h}^2 \label{proof3},
\end{align} 
and the last equality holds by Remark~\ref{ch4_projector}(iii).
For the second term of the right-hand side of \eqref{proof2} we use Theorem~\ref{theorem4.12}.
We finish the proof of both (i) and (ii) by plugging  \eqref{proof2} and \eqref{proof3} into \eqref{proof1}.
\eproof

{Since the prediction error decreases with $d$, Theorem~\ref{theorembound} can not be applied as a criterion for the choice of $d$. 
In a data analysis, when quantities such as covariance operators and its eigenvalues have to be estimated, the variance of the estimators increases with $d$. 
Small errors in the estimation of small empirical eigenvalues may have severe consequences on the prediction error (see \citet{bernard}). 
{ To avoid this problem  a conservative choice of $d$ is suggested.}
Theorem~\ref{theorembound} allows for an interpretation of the prediction error for fixed $d$. This is similar as for Theorem~3.2 in \citet{aue}, here for a more general model class of ARMA models. }

\section{Traffic data analysis}\label{s6}

\begin{figure}[t]
\begin{center}
\includegraphics[trim=1.2cm 1.5cm 1.05cm 1.5cm, clip=true,scale=0.67]{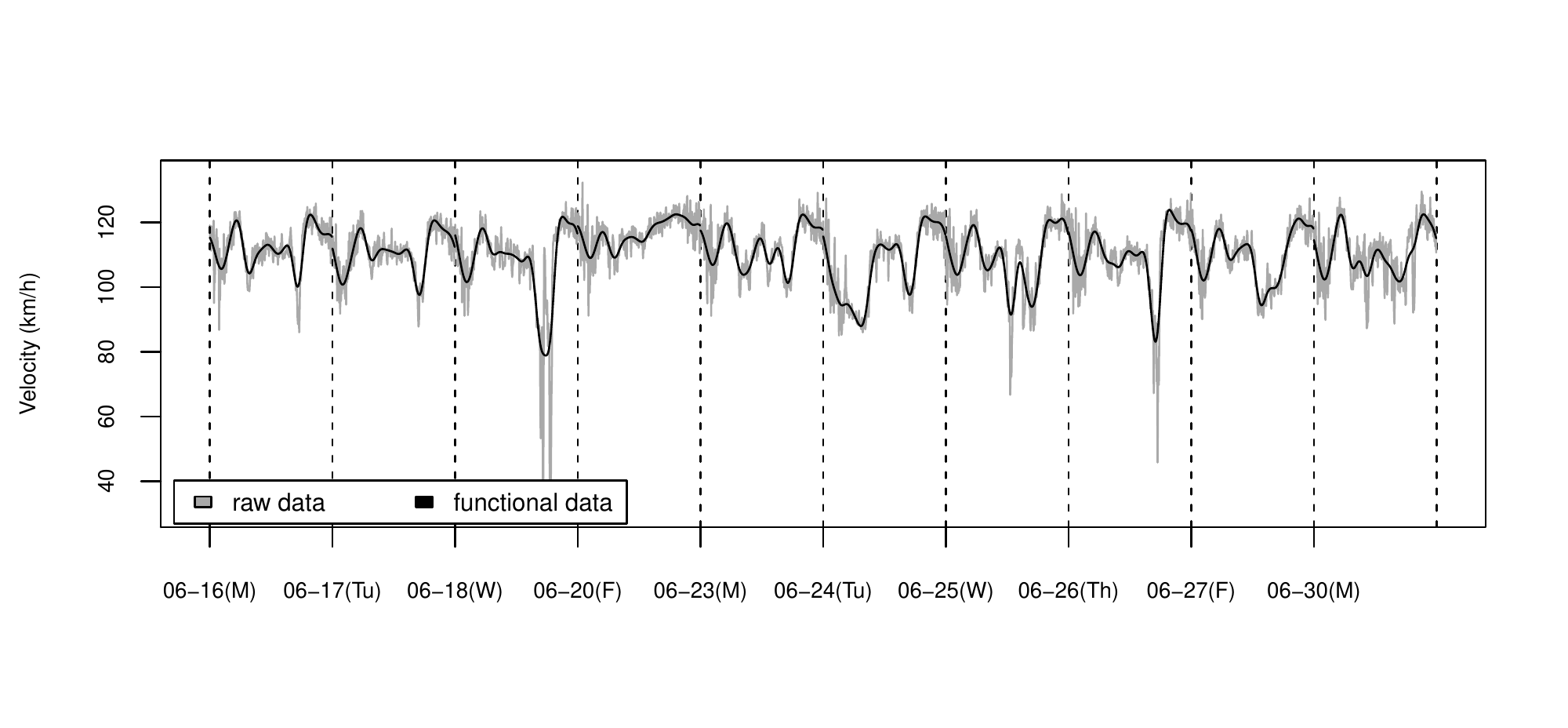}
  \caption{Functional velocity data (black) and raw data (grey) both in km/h on the last ten working days in June 2014 (June 19th 2014 was a catholic holiday).}
  \label{raw_fd_last10days_workingdays}
\end{center}
\end{figure}
We apply the functional time series prediction theory of Section~4 to highway traffic data provided by the Autobahndirektion S\"udbayern, thus extending previous work in \citet{bessecardot}.
Our dataset consists of measurements at a fixed point on a highway (A92) in Southern Bavaria, Germany.
Recorded is the average velocity per minute from 1/1/2014 00:00 to 30/06/2014 23:59 on three lanes.
After taking care of missing values and data outliers, we average the velocity per minute over the three lanes, weighted by the number of vehicles per lane.
Then we transform the cleaned daily high-dimensional data to functional data using the first $30$ Fourier basis functions. 
The two standard bases of function spaces used in FDA are Fourier and B-spline basis functions (see \citet{ramsay1}, Section~3.3). 
We choose Fourier basis functions as they allow for a more parsimonious representation of the variability: a Fourier representation needs only 4 FPCs to explain 80\% of the variability in the data,  whereas a B-spline representation requires 6 (see \citet{taoran}, Section~6.1).
In Figure~\ref{raw_fd_last10days_workingdays} we depict the outcome on the working days of two weeks in June 2014.
More information on the transformation from discrete time observation to functional data and details on the implementation in \texttt{R} are provided in \cite{taoran}, Chapter~6.

\begin{figure}[t]
\begin{center}
\includegraphics[trim=1.2cm 1.6cm 1cm 2cm, clip=true,scale=0.67]{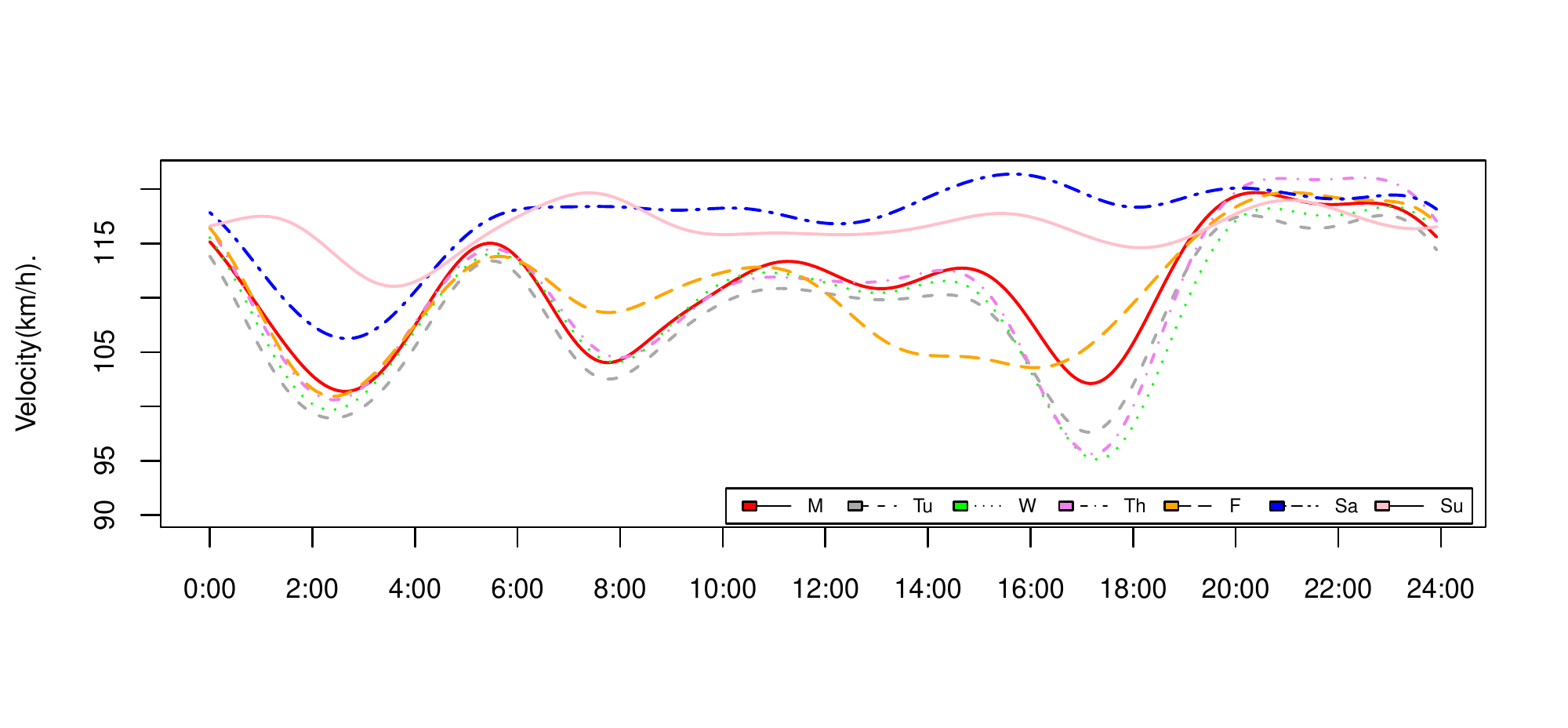}
\caption{Empirical functional mean velocity {(in km/h) on the 7 days of the week, over the day}}
\label{weekday_mean}
\end{center}
\end{figure}

\begin{figure}[t]
\begin{center}
\includegraphics[trim=1.2cm 1.6cm 1cm 2cm, clip=true,scale=0.72]{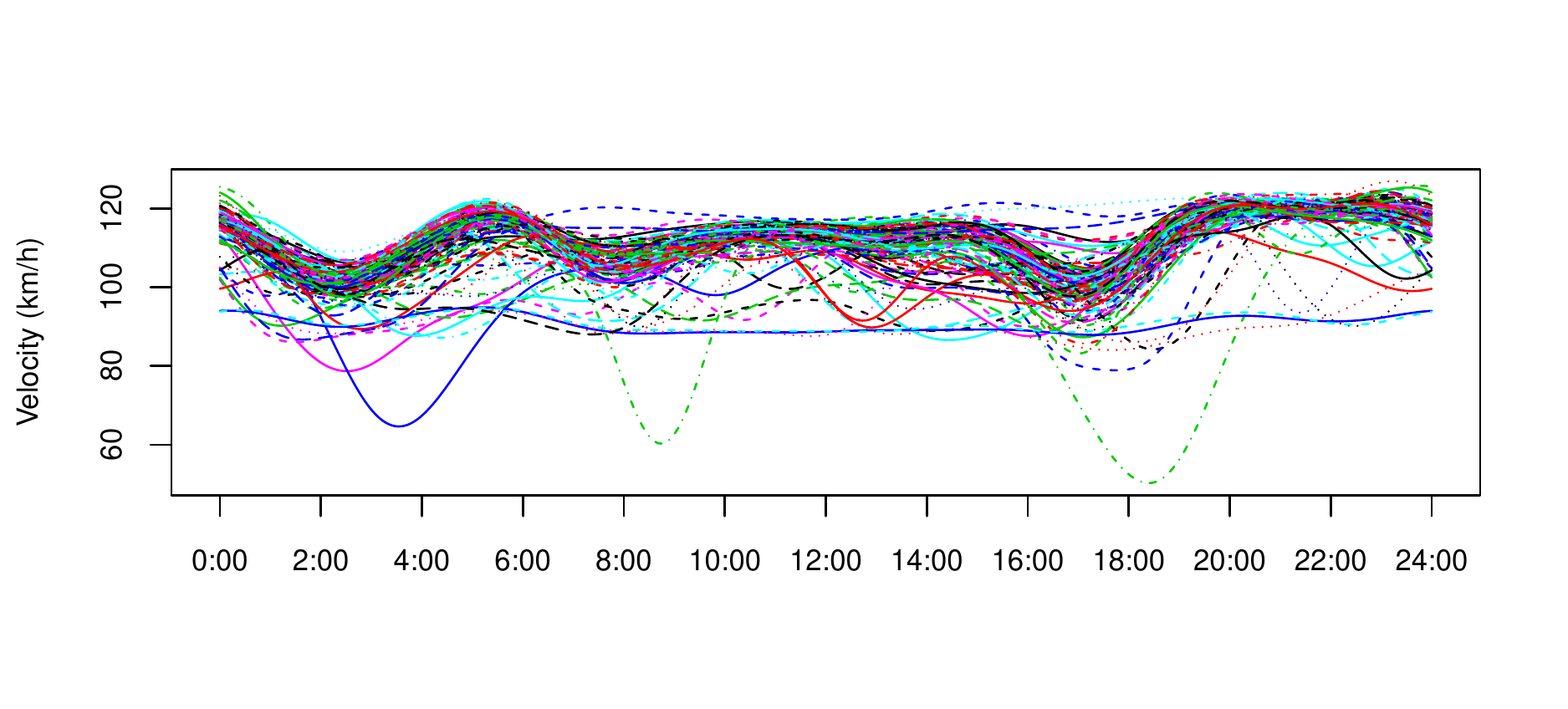}\\
\includegraphics[trim=1.2cm 1.5cm 1cm 2cm, clip=true,scale=0.72]{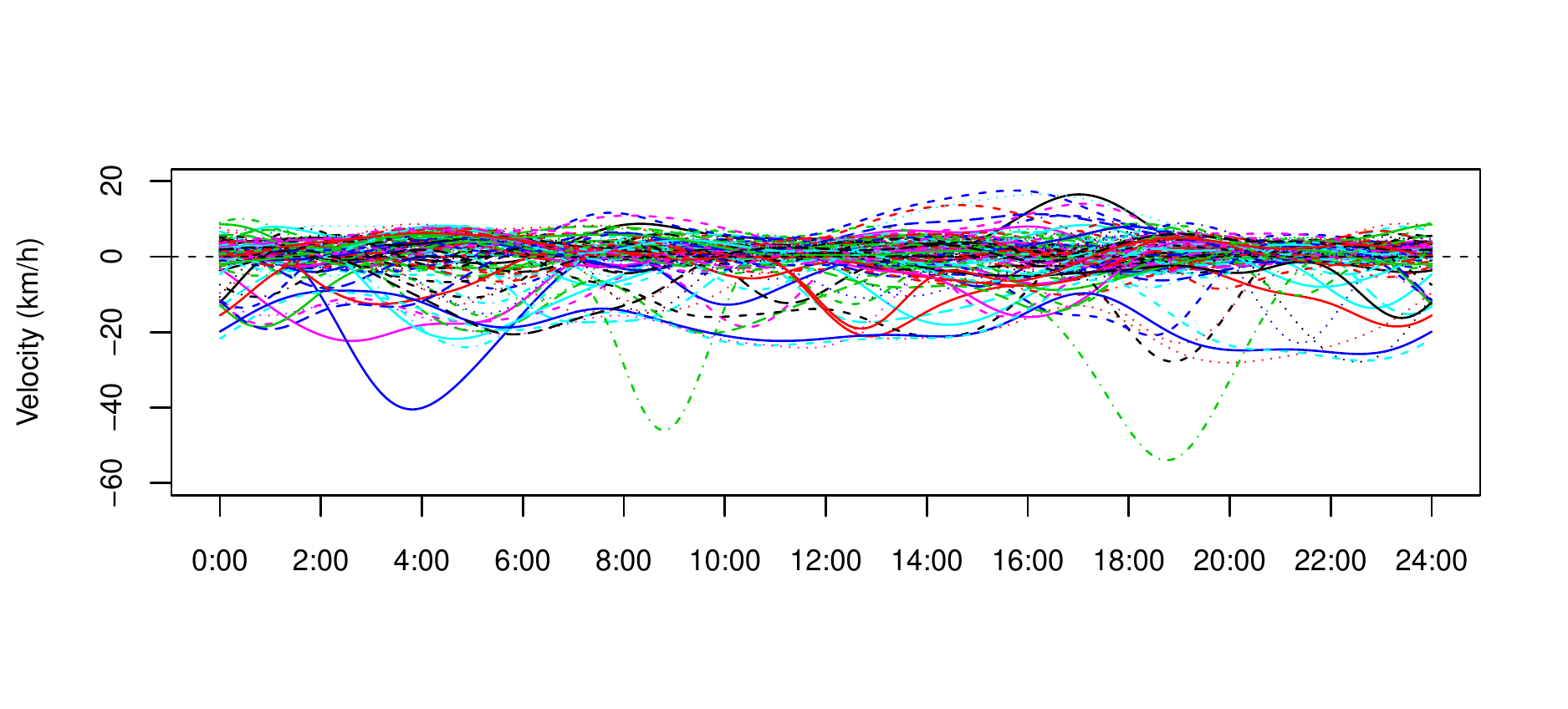}
  \caption{Functional velocity data (in km/h) over the day for 30 working days smoothed by a Fourier basis before and after subtracting the weekday mean}
  \label{Meansubstraction_indiv}
\end{center}
\end{figure}

As can be seen in Figure~\ref{weekday_mean}, different weekdays have different mean velocity functions. 
To account for the difference between weekdays we subtract the empirical individual daily mean from all daily data (Monday mean from Monday data, etc.). The effect is clearly visible in Figure~\ref{Meansubstraction_indiv}. 
However, even after deduction of the daily mean, the functional stationarity test of \citet{rice} based on projection rejects stationarity of the time series. This is due to traffic flow on weekends: Saturday and Sunday traffic show different patterns than weekdays, even after mean correction.
Consequently, we restrict our investigation to working days (Monday-Friday), resulting in a functional time series $X_n$ for $n=1,\ldots,N=119$, for which the stationarity test suggested in \cite{rice} does not reject the stationarity assumption.

A portmanteau test of \citet{gabrys} applied to $X_n$ for $n=1,\dots,N=119$ working days rejects (with a $p$-value as small as $10^{-6}$) that the daily functional data are uncorrelated. 
{The assumption of temporal dependence in the data is in line with the results in \citet{chrobok} who use linear models to predict inner city traffic flow, and with results 
in \citet{bessecardot} who use the temporal dependence for the prediction of traffic volume with a functional AR$(1)$ model}.

Next we show the prediction method at work for our data. 
More precisely, we estimate the $1$-step predictors for the last $10$ working days of our dataset and present the final result in Figure~\ref{functionalpredictor_workingdays}, where we compare the functional velocity data with their 1-step predictor.
We  explain the procedure in detail.

We start by estimating the covariance operator (recall Remark~\ref{consistent}). Figure~\ref{ch2_figcovkernel} shows the empirical covariance kernel of the  functional traffic velocity data based on $119$ working days (the empirical version of $E[(X(t)-\mu(t))(X(s) - \mu(s))]$ for $0\leq t,s \leq 1$). 

\begin{figure}[t]
\begin{center}
\includegraphics[trim=3cm 2cm 0.4cm 2.5cm, clip=true,scale=0.6]{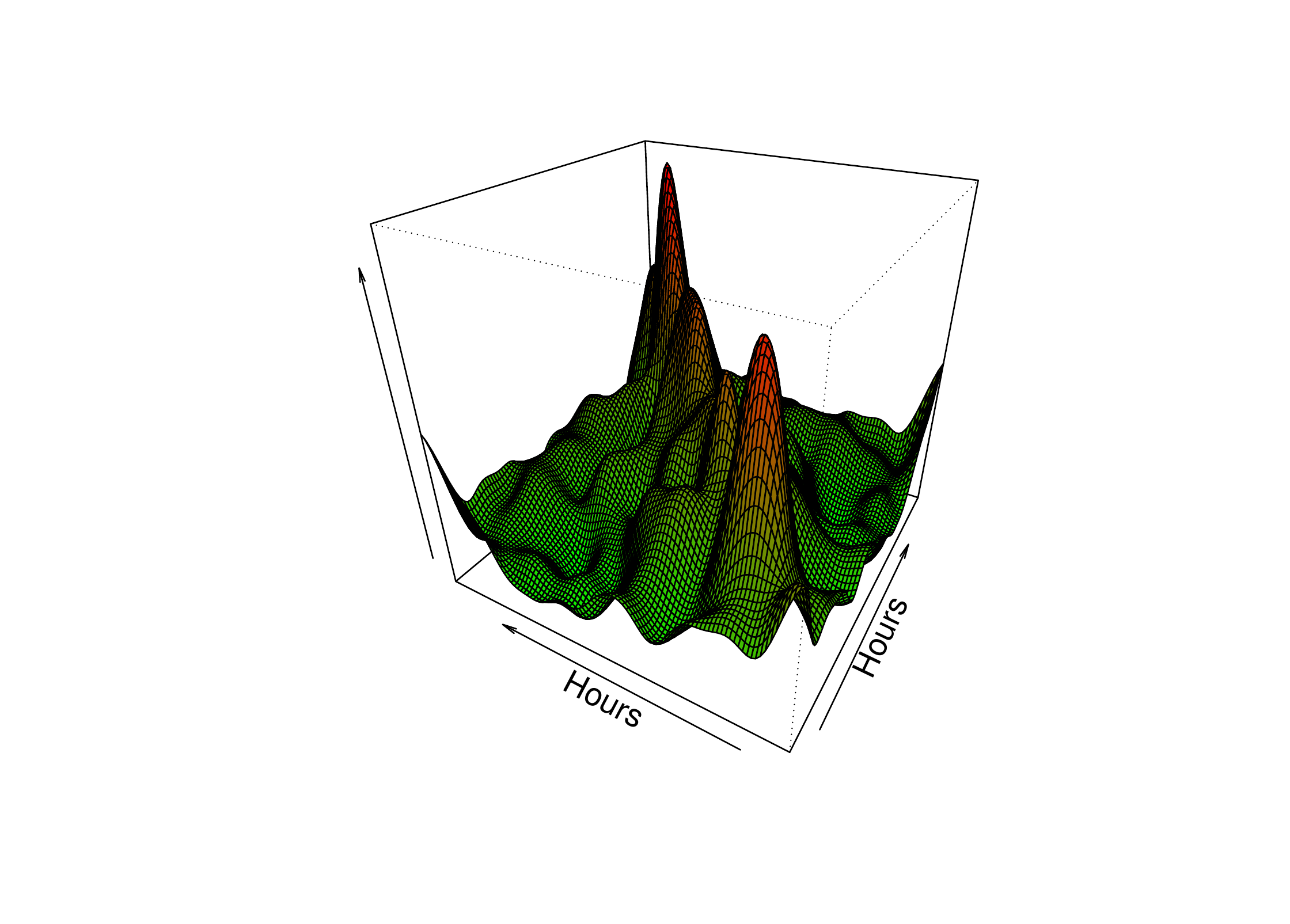}
  \caption{Empirical covariance kernel of functional velocity data on $119$ working days.}
  \label{ch2_figcovkernel}
  \end{center}
\end{figure}

As indicated by the arrows, the point $(t,s) =( 0,0)$ is at the
bottom right corner and estimates the variance at midnight. 
The empirical variance over the day is represented along the diagonal from the bottom right to the top left corner. 
The valleys and peaks along the diagonal represent phases of low and high traffic density: for instance, the first peak represents the variance at around
05:00 a.m., where traffic becomes denser, since commuting to work starts. 
Peaks away from the diagonal represent high dependencies between different time points during the day. 
For instance, high traffic density in the early morning correlates with high traffic density in the late afternoon, again due to commuting.

\begin{figure}[t]
\begin{center}
 \includegraphics[trim=1.2cm 1.5cm 0.8cm 2cm, clip=true,scale=0.66]{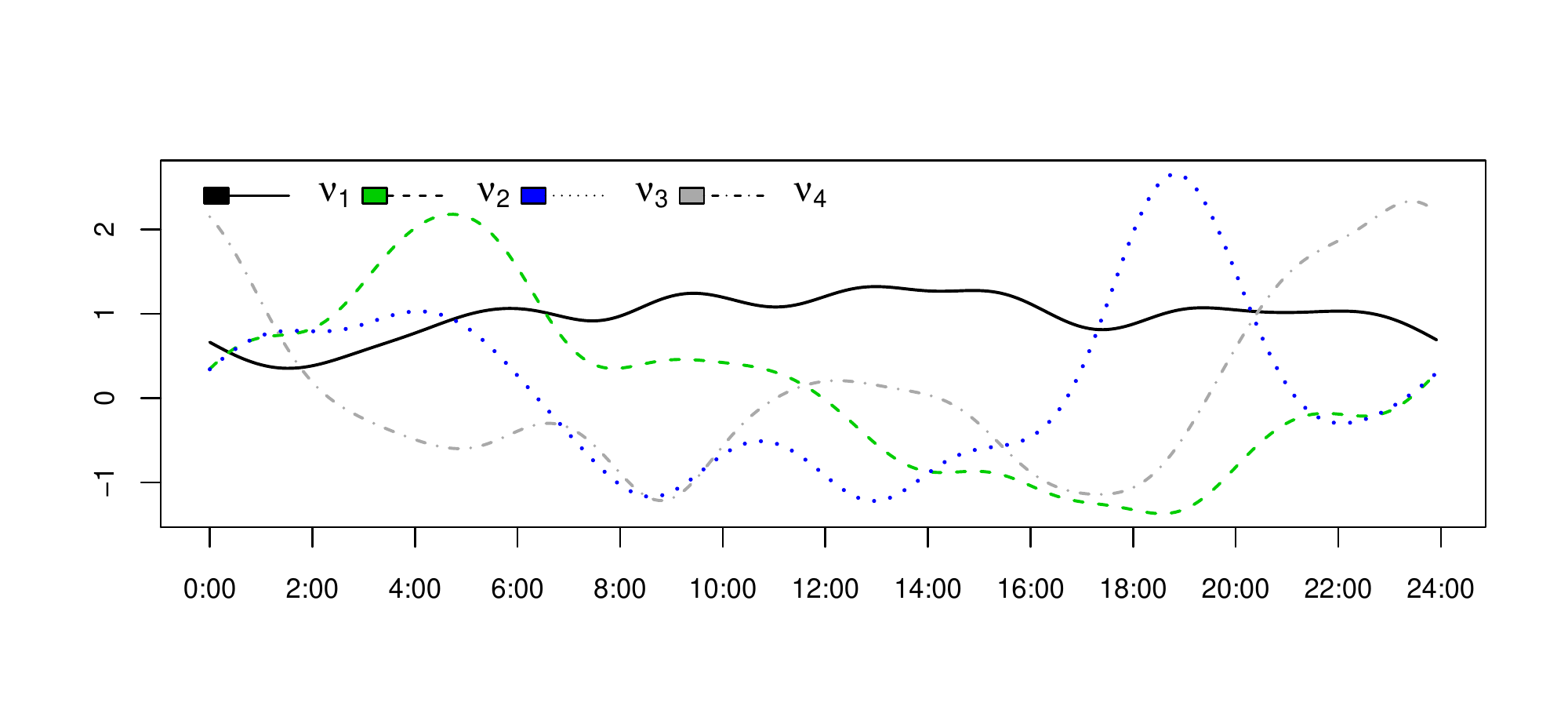}
  \caption{Four empirical eigenfunctions of the $N=119$ working days functional velocity data. 
  The criterion is 80\%; i.e., $\nu_1,\nu_2,\nu_3,\nu_4$ explain together 80\% of the total data variability.
  \label{fig_eigenfunctions}}
  \end{center}
\end{figure}

Next we compute the empirical eigenpairs $(\lambda_j^e,\nu_j^e)$ for $j=1,\ldots,N$ of the empirical covariance operator. The first four eigenfunctions are depicted in Figure~\ref{fig_eigenfunctions}.

Now we apply the \CPV\ method from Remark~\ref{CVP} to the  functional highway velocity data. 
From a ``$\CPV(d)$ vs. $d$'' plot we read off that $d=4$ FPCs explain 80\% of the variability of the data. 

Obviously, the choice of $d$ is critical. Choosing $d$ too small induces a loss of information as seen in Theorem~4.13. 
Choosing $d$ too large makes the estimation of the vector model difficult and may result in imprecise predictors: the prediction error may explode (see \citet{bernard}). As a remedy we perform cross validation on the prediction error based on a different number $d$ of relevant scores. 
This furthermore ensures that the dependence structure of the data is not ignored when it is relevant for prediction.

Since the prediction is not only based on the number of scores, but also on the chosen ARMA model, we perform cross validation on the number of scores in combination with cross validation on the orders of the ARMA models.

Thus, we apply the Algorithm~1 of Section~4.1 to the functional velocity data and implement the following steps for $d=2,\dots,6$ and $N=119$.

(1) For each day $n\in\{1,\dots,N\}$, truncate the Karhunen-Lo\`eve representation (Theorem~\ref{theorem2.3}) of the daily functional velocity curve $X_n$ at $d$. 
This yields
$$X_{n,d}:=\sum\limits_{j=1}^{d}\left\langle X_n,\nu^e_j\right\rangle\nu^e_j, \quad n=1,\dots,N.$$ 
(Figure~\ref{ch1_rawdatavsfd} depicts the (centered) functional velocity data and the corresponding truncated data for $d=4$.)
Store the $d$ scores in the vector $\mathbf{X}_n$, 
$$\bfx_n=(\left<X_n,\nu^e_1\right>,\dots,\left<X_n,\nu^e_d\right>)^\top, \quad n=1,\dots,N.$$

\begin{figure}[t]
\begin{center}
\includegraphics[trim=1.2cm 1.6cm 0.4cm 2cm, clip=true,scale=0.645]{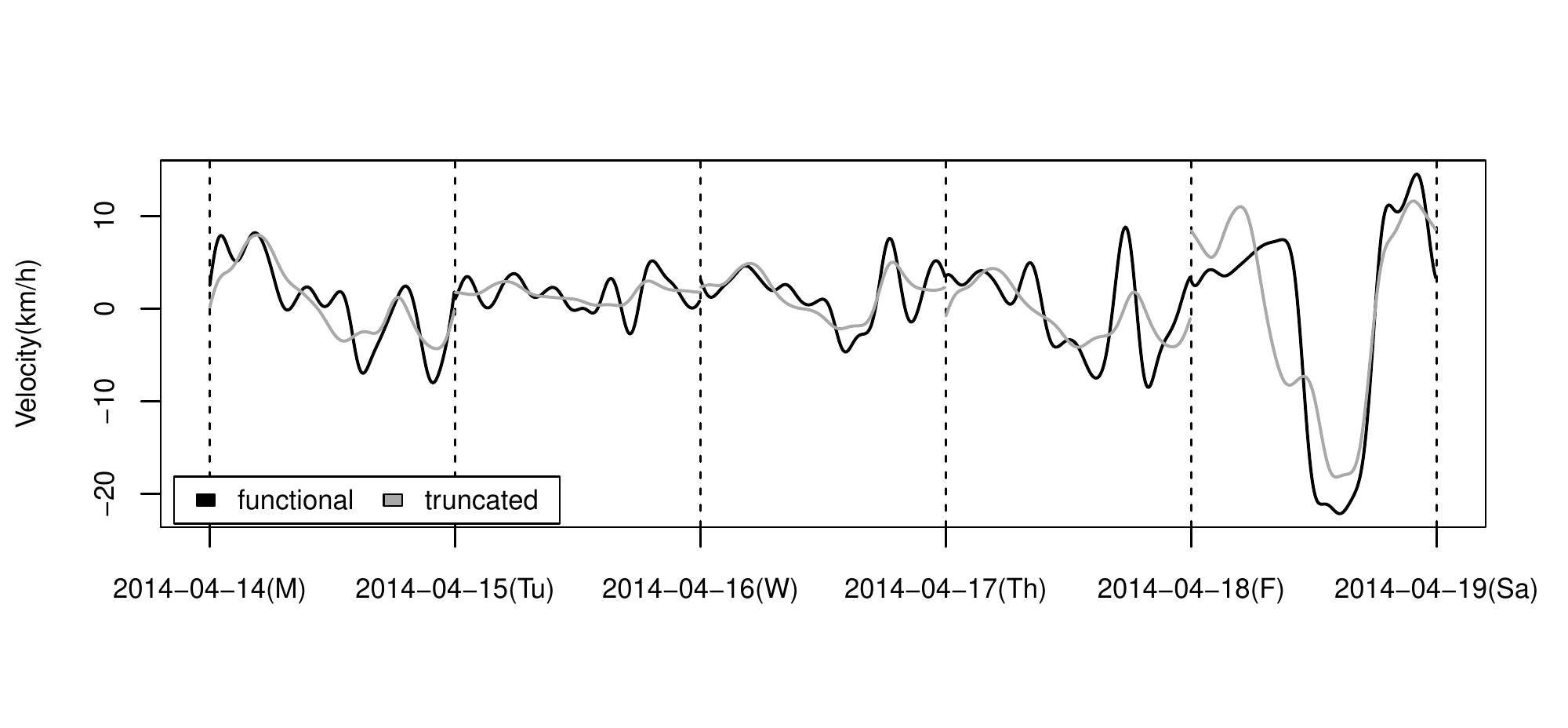}
 \caption{Functional velocity raw data on 5 consecutive working days (black) versus the truncated data by the Karhunen-Lo\`eve representation (grey). The criterion is 80\% and the number $d$ of FPCs is 4.}
 \label{ch1_rawdatavsfd}
 \end{center}
\end{figure}

(2) {Fit different vector ARMA$(p,q)$ models to the $d$-dimensional score vector. 
Compute the best linear predictor $\wh\bfx_{n+1}$ based on the vector model iteratively by  the Durbin-Levinson or the Innovations Algorithm (see e.g \citet{brockwell}). 

(3) Re-transform the vector best linear predictor $\wh\bfx_{n+1}$ into its functional form $\wh X_{n+1}$. 
Compare the goodness of fit of the models by their functional prediction errors $\Vert X_{n+1} - \wh X_{n+1}\Vert^2$. (In
Table~\ref{ch6_predictionworkingdays} root mean squared errors (RMSE) and mean absolute errors (MAE)
for the different models are summarized.) 
Fix the optimal $d$ and the optimal
ARMA$(p, q)$ model.

{As a result we find minimal 1-step prediction errors for $d=4$, which confirms the choice proposed by the CPV method, and for the VAR$(2)$ and the vector MA$(1)$ model.
Both models yield the same RMSE, and the MAE of the vector MA$(1)$ model is slightly smaller than that of the vector AR$(2)$ model.} Since we opt for a parsimonious model, we choose the vector MA$(1)$ model, for which the predictor is depicted in Figure~\ref{functionalpredictor_workingdays}.}

\begin{table}[t]
\centering
\begin{small}
\begin{tabular}{c c c  c  c c c c c  }
  \hline                  
& $(p,q)$ & $(1,0)$& $(2,0)$& $(0,1)$& $(0,2)$ & $(1,1)$ & $(2,1)$ & $(1,2)$ \\ \hline
d=2& RMSE & 5.15 &	5.09 & 5.02 &	5.15 &	5.13 &	4.96 &	5.09 \\
   & MAE  & 3.82 &  3.77 & 3.73 &	3.83 &	3.80 &	3.66 &	3.76 \\ \hline
d=3& RMSE & 4.97 &	4.87 & 4.86 &	5.30 &	4.94 &	4.89 &	5.08 \\
   & MAE  & 3.70 &	3.62 & 3.61	&   3.87 &	3.68 &	3.63 &	3.69 \\ \hline
d=4& RMSE & 4.98 &	\textbf{4.83} &	\textbf{4.83}&	5.55 & 4.92 & 4.90 & 5.23 \\
   & MAE  & 3.67 &	3.55 &	\textbf{3.54}&	4.13 &	3.62 &	3.61 &	3.83 \\ \hline
d=5& RMSE & 5.06 &	5.15 &	4.91&	5.80 &	5.04 &  5.20 &  5.46   \\
   & MAE  & 3.76 &	3.77 &	3.63&	4.38 &	3.76 &  3.80 &  4.02\\ \hline 
d=6& RMSE & 5.12 &	5.28 &	5.09&	6.47 &	5.12 &  5.34 &	5.97     \\
   & MAE  & 3.78 &	3.88 &	3.82 &	4.87 &  3.81 &	3.91 &  4.50  \\ \hline
 \end{tabular}
 \end{small}
 \caption{Average 1-step prediction errors of the predictors for the last 10 observations for all {working days} for different \ARMA\ models and number of principal components.}
 \label{ch6_predictionworkingdays}
\end{table}

\begin{figure}[t]
\begin{center}
\includegraphics[trim=1.2cm 1.6cm 1cm 2cm, clip=true,scale=0.67]{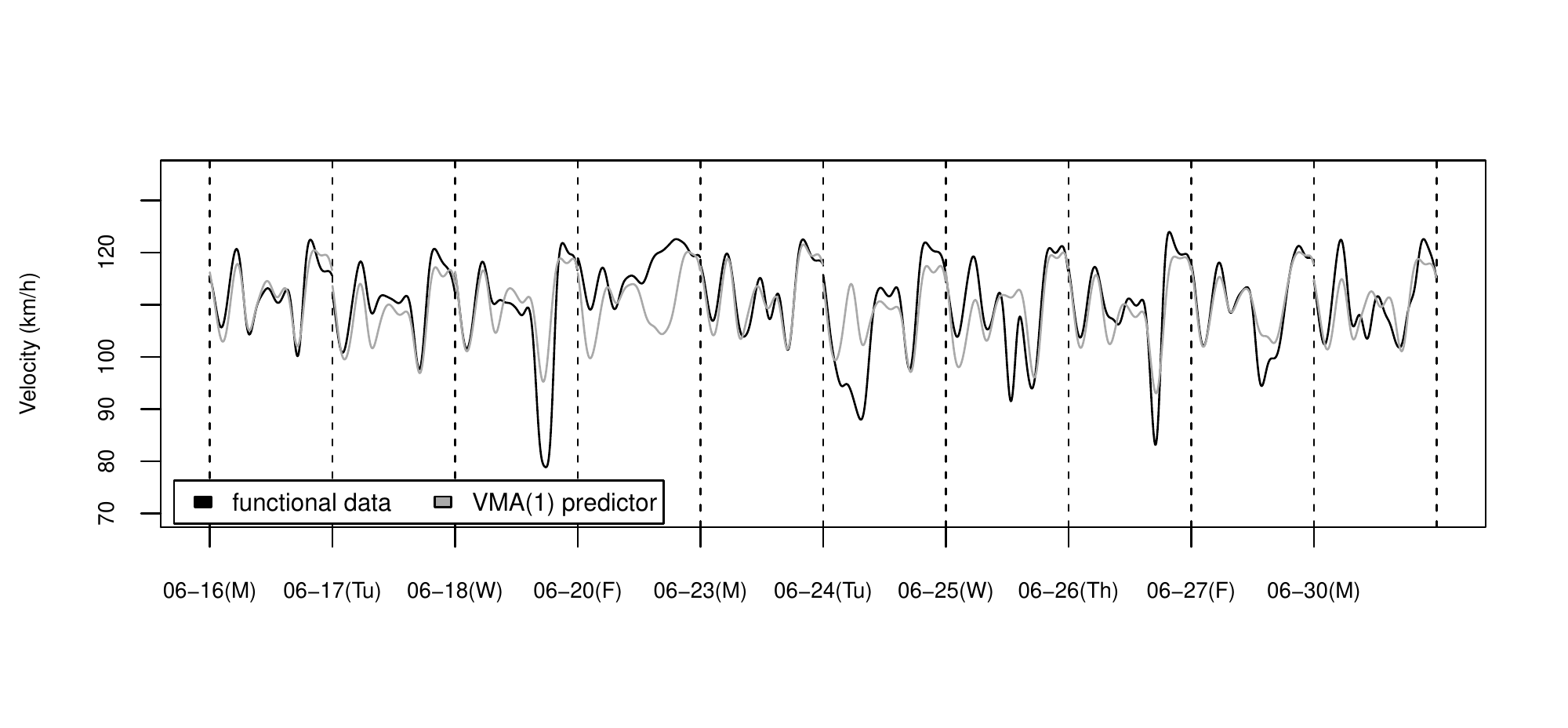}
  \caption{Functional velocity data in black and 1-step functional predictor based on VMA$(1)$ in grey (in km/h) for the last 10 working days in June 2014 }
  \label{functionalpredictor_workingdays}
  \end{center}
\end{figure}

{Finally, we compare the performance of the 1-step prediction based on the functional MA$(1)$ model with standard non-parametric prediction methods. 
Our approach definitely outperforms prediction methods like exponential smoothing, 
naive prediction with the last observation, 
or using the mean of the time series as a predictor. 
Details are given in \citet{taoran}, Section~6.3.}

\section{Conclusions}

{We have investigated functional {ARMA$(p,q)$} models and a corresponding approximating vector model, which lives on the closed linear span of the first $d$ eigenfunctions of the covariance operator. 
We have presented conditions for the existence of a unique stationary  and causal solution to both functional ARMA$(p,q)$ and approximating vector model.
Furthermore, we have derived conditions such that the approximating vector model is exact. 
Interestingly, and in contrast to AR or ARMA models, for a functional MA process of finite order the approximate vector process is automatically again a MA process of equal or smaller order.

For arbitrary $h\in\N$ we have investigated the $h$-step functional best linear predictor of \citet{bosq2014} and gave conditions for a representation in terms of operators in $\call$.
We have compared the best linear predictor of the approximating vector model with the functional best linear predictor,  and showed that the difference between the two predictors tends to $0$ if the dimension of the vector model $d\rightarrow\infty$.
The theory gives rise to a prediction methodology for stationary functional ARMA$(p,q)$ processes similar to the one introduced in \citet{aue}. 

We have applied the new prediction theory to traffic velocity data. 
For finding an appropriate dimension $d$ of the vector model, we applied
the FPC criterion and cross validation on the prediction error.
For our traffic data the cross validation leads to the same choice of $d=4$ as the FPC criterion for CPV$(d)\geq 80\%$.
The model selection is also performed via cross validation on the $1$-step prediction error for different ARMA models resulting in a MA$(1)$ model.

The appeal of the methodology is its ease of application. 
Well-known \texttt{R} software packages (\texttt{fda} and \texttt{mts}) make the implementation straightforward. 
Furthermore, the generality of dependence induced by ARMA models extends the range of application of functional time series, which was so far restricted to autoregressive dependence structures.\\ }

\FloatBarrier
\noindent\textbf{Acknowledgement:} We thank the Autobahndirektion S\"udbayern and especially J.~Gr\"otsch for their support and for providing the traffic data. J.~Klepsch furthermore acknowledges financial support from the Munich Center for Technology and Society based project ASHAD. 

\section*{References}
\bibliography{KKWbibliography}
\end{document}